\documentclass[smallextended]{svjour3}     
\smartqed  
\usepackage{hyperref}
\usepackage{graphicx}
\usepackage{mathptmx}      
\usepackage{latexsym}
\usepackage{graphicx, amssymb, amsfonts, amsmath, float,color}
\usepackage{appendix}

\newcommand{\cbr}{C_{\mathrm{bro}}}
\newcommand{\calv}{C_{\mathrm{A}}}
\newcommand{\cart}{C_{\mathrm{a}}}
\newcommand{\cven}{C_{\bar{\mathrm{v}}}}
\newcommand{\cliv}{C_{\mathrm{liv}}}
\newcommand{\ctis}{C_{\mathrm{tis}}}
\newcommand{\cinh}{C_{\mathrm{I}}}
\newcommand{\cmeas}{C_{\mathrm{measured}}}

\newcommand{\cw}{C_{\mathrm{water}}}
\newcommand{\coxy}{C_{\mathrm{O}_2}}
\newcommand{\ccd}{C_{\mathrm{CO}_2}}

\newcommand{\lliv}{\lambda_{\mathrm{b:liv}}}
\newcommand{\ltis}{\lambda_{\mathrm{b:tis}}}
\newcommand{\hen}{\lambda_{\mathrm{b:air}}}
\newcommand{\lmuca}{\lambda_{\mathrm{muc:air}}}
\newcommand{\lmucb}{\lambda_{\mathrm{muc:b}}}

\newcommand{\qbr}{q_{\mathrm{bro}}}
\newcommand{\qliv}{q_{\mathrm{liv}}}
\newcommand{\qalv}{\dot{V}_{\mathrm{A}}}
\newcommand{\qc}{\dot{Q}_{\mathrm{c}}}

\newcommand{\pr}{k_\mathrm{pr}}
\newcommand{\met}{k_{\mathrm{met}}}
\newcommand{\dke}{k_{\mathrm{diff},1}}
\newcommand{\dkz}{k_{\mathrm{diff},2}}
\newcommand{\klin}{k_{\mathrm{lin}}}
\newcommand{\vmax}{v_\mathrm{max}}
\newcommand{\km}{k_\mathrm{m}}

\newcommand{\vbro}{\tilde{V}_{\mathrm{bro}}}
\newcommand{\valv}{\tilde{V}_{\mathrm{A}}}
\newcommand{\vt}{V_{\mathrm{T}}}
\newcommand{\vtis}{\tilde{V}_{\mathrm{tis}}}
\newcommand{\vliv}{\tilde{V}_{\mathrm{liv}}}

\newcommand{\tb}{\vec{p}}
\newcommand{\bb}{\vec{b}}
\newcommand{\cb}{\vec{c}}

\newcommand{\gb}{\vec{g}}
\newcommand{\xb}{\vec{x}}
\newcommand{\ab}{\vec{\alpha}}

\newcommand{\zb}{\vec{z}}
\newcommand{\ub}{\vec{u}}

\newcommand{\di}{\mathrm{d}}
\newcommand{\R}{\mathbb{R}}

\newcommand{\cinf}{\mathcal{C}^{\infty}}
\newcommand{\bw}{\mathrm{bw}}

\newcommand{\lk}{\left(}
\newcommand{\rk}{\right)}

\newcommand{\norm}[1]{\lVert #1 \rVert}
\newcommand{\abs}[1]{| #1 |}

 \journalname{Journal of Mathematical Biology}

\begin{document}

\title{A mathematical model for breath gas analysis of volatile organic compounds with special emphasis on acetone}

\titlerunning{A mathematical model for breath gas analysis of VOCs}        

\author{Julian~King   \and
        Karl~Unterkofler  \and
        Gerald~Teschl \and
        Susanne~Teschl  \and
        Helin~Koc \and
        Hartmann~Hinterhuber  \and
        Anton~Amann}

\authorrunning{King, Unterkofler, Teschl, Teschl, Koc, Hinterhuber, Amann}  

\institute{
J.~King \at
Breath Research Institute of the Austrian Academy of Sciences\\
Rathausplatz~4, A-6850 Dornbirn, Austria\\
\email{julian.king@oeaw.ac.at}  
\and
K.~Unterkofler \at
Vorarlberg University of Applied Sciences and\\
Breath Research Unit of the Austrian Academy of Sciences\\
Hochschulstr.~1, A-6850 Dornbirn, Austria\\
\email{karl.unterkofler@fhv.at}          
\and
G.~Teschl \at
University of Vienna, Faculty of Mathematics\\
Nordbergstr.~15, A-1090 Wien, Austria\\
\email{gerald.teschl@univie.ac.at}  
\and
S.~Teschl \at
University of Applied Sciences Technikum Wien\\
H\"ochst\"adtplatz~5, A-1200 Wien, Austria\\
\email{susanne.teschl@technikum-wien.at}
\and
H.~Koc \at
Vorarlberg University of Applied Sciences \\
Hochschulstr.~1, A-6850 Dornbirn, Austria
 \and
H.~Hinterhuber \at
Innsbruck Medical University, Department of Psychiatry\\
Anichstr.~35, A-6020 Innsbruck, Austria
\and
A.~Amann (corresponding author) \at
Univ.-Clinic for Anesthesia, Innsbruck Medical University and\\
Breath Research Institute of the Austrian Academy of Sciences\\
Anichstr.~35, A-6020 Innsbruck, Austria\\
\email{anton.amann@oeaw.ac.at, anton.amann@i-med.ac.at} 
}

\date{Received: 23 March 2010 / Revised: 9 November 2010}

\maketitle
\newpage
\begin{abstract}
Recommended standardized procedures for determining exhaled lower respiratory nitric oxide and nasal nitric oxide (NO) have been developed by task forces of the European Respiratory Society and the American Thoracic Society. These recommendations have paved the way for the measurement of nitric oxide to become a diagnostic tool for specific clinical applications. It would be desirable to develop similar guidelines for the sampling of other trace gases in exhaled breath, especially volatile organic compounds (VOCs) which may reflect ongoing metabolism.

The concentrations of water-soluble, blood-borne substances in exhaled breath are influenced by:
\begin{itemize}
\item breathing patterns affecting gas exchange in the conducting airways
\item the concentrations in the tracheo-bronchial lining fluid
\item the alveolar and systemic concentrations of the compound.
\end{itemize}
The classical Farhi equation takes only the alveolar concentrations into account. Real-time measurements of acetone in end-tidal breath under an ergometer challenge show characteristics which cannot be explained within the Farhi setting. Here we develop a compartment model that reliably captures these profiles and is capable of relating breath to the systemic concentrations of acetone. By comparison with experimental data it is inferred that the major part of variability in breath acetone concentrations (e.g., in response to moderate exercise or altered breathing patterns) can be attributed to airway gas exchange, with minimal changes of the underlying blood and tissue concentrations. Moreover, the model illuminates the discrepancies between observed and theoretically predicted blood-breath ratios of acetone during resting conditions, i.e., in steady state. Particularly, the current formulation includes the classical Farhi and the Scheid series inhomogeneity model as special limiting cases and thus is expected to have general relevance for a wider range of blood-borne inert gases.

The chief intention of the present modeling study is to provide mechanistic relationships for further investigating the exhalation kinetics of acetone and other water-soluble species. This quantitative approach is a first step towards new guidelines for breath gas analyses of volatile organic compounds, similar to those for nitric oxide.

\end{abstract}
\newpage
%
\section{Introduction}
\label{sect:intro}
%
Measurement of blood-borne volatile organic compounds (VOCs) occurring in human exhaled breath as a result of normal metabolic activity or pathological disorders has emerged as a promising novel methodology for non-invasive medical diagnosis and therapeutic monitoring of disease, drug testing and tracking of physiological processes~\cite{amannbook,amann2007,amann2004,rieder2001,miekisch2006}. Apart from the obvious improvement in patient compliance and tolerability, major advantages of exhaled breath analysis compared to conventional test procedures, e.g., based on blood or urine probes, include de facto unlimited availability as well as rapid on-the-spot evaluation or even \emph{real-time} analysis. Additionally, it has been pointed out that the pulmonary circulation receives the entire cardiac output and therefore the breath concentrations of such compounds might provide a more faithful estimate of pooled systemic concentrations than single small-volume blood samples, which will always be affected by local hemodynamics and blood-tissue interactions~\cite{ohara2008}. \\ Despite this huge potential, the use of exhaled breath analysis within a clinical setting is still rather limited. This is mainly due to the fact that drawing reproducible breath samples remains an intricate task that has not fully been standardized yet. Moreover, inherent error sources introduced by the complex mechanisms driving pulmonary gas exchange are still poorly understood. The lack of standardization among the different sampling protocols proposed in the literature has led to the development of various sophisticated sampling systems, which selectively extract end-tidal air by discarding anatomical dead space volume~\cite{setup,herbig2008,birken2006}). Even though such setups present some progress, they are far from being perfect. In particular, these sampling systems can usually not account for the variability stemming from varying physiological states.\\

In common measurement practice it is often tacitly assumed that end-tidal air will reflect the alveolar concentration $\calv$, which in turn is proportional to the concentration of the VOC in mixed venous blood $\cven$, with the associated factor depending on the substance-specific blood:gas partition coefficient $\hen$ (describing the diffusion equilibrium between capillaries and alveoli), alveolar ventilation $\qalv$ (governing the transport of the compound through the respiratory tree) and cardiac output $\qc$ (controlling the rate at which the VOC is delivered to the lungs):
\begin{equation}\label{eq:farhi}\cmeas=\calv=\frac{\cven}{\hen+\frac{\qalv}{\qc}}.\end{equation}
This is the familiar equation introduced by Farhi~\cite{farhi1967}, describing steady state inert gas elimination from the lung viewed as a single alveolar compartment with a fixed overall ventilation-perfusion ratio $\qalv/\qc$ close to one.
Since the pioneering work of Farhi, both equalities in the above relation have been challenged. Firstly, for low blood soluble inert gases, characterized by $\hen \leq 10$, alveolar concentrations resulting from an actually constant $\cven$ can easily be seen to vary drastically in response to fluctuations in blood or respiratory flow (see also~\cite{King2010GC,isoprene} for some recent findings in this context). While this sensitivity has been exploited in MIGET (Multiple Inert Gas Elimination Technique, cf.~\cite{wagner1974,wagner2008}) to assess ventilation-perfusion inhomogeneity throughout the normal and diseased lung, it is clearly problematic in standard breath sampling routines based on free breathing, as slightly changing measurement conditions (regarding, e.g., body posture, breathing patterns or stress) can have a large impact on the observed breath concentration~\cite{cope2004}. This constitutes a typical example of an inherent error source as stated above, potentially leading to a high degree of intra- and consequently inter-individual variability among the measurement results~\cite{ohara2009}. It is hence important to investigate the influence of medical parameters like cardiac output, pulse, breathing rate and breathing volume on VOC concentrations in exhaled breath. 
In contrast, while highly soluble VOCs ($\hen>10$) tend to be less affected by changes in ventilation and perfusion, measurement artifacts associated with this class of compounds result from the fact that -- due to their often hydrophilic properties -- a substantial interaction between the exhalate and the mucosa layers lining the conducting airways can be anticipated~\cite{anderson2007}. In other words, for these substances $\cmeas \neq \calv$, with the exact quantitative relationship being unknown. Examples of endogenous compounds that are released into the gas phase not only through the blood-alveolar interface, but also through the bronchial lining fluid are, e.g., acetone and ethanol~\cite{anderson2006,tsu1988}.\\

Acetone (2-propanone; CAS number 67--64--1; molar mass $58.08$~g/mol) is one of the most abundant VOCs found in human breath and has received wide attention in the biomedical literature. Being a natural metabolic intermediate of lipolysis~\cite{kalapos2003}, endogenous acetone has been considered as a biomarker for monitoring the ketotic state of diabetic and fasting individuals~\cite{tassopoulos1969,owen1982,reichard1979,smith1999}, estimating glucose levels~\cite{galassetti2005} or assessing fat loss~\cite{kundu1993}. Nominal levels in breath and blood have been established in~\cite{wangace,schwarzace}, and bioaccumulation has been studied in the framework of exposure studies and pharmacokinetic modeling~\cite{wigaeus1981,kumagai2000,moerk2006}. 

Despite this relatively large body of experimental evidence, the crucial link between acetone levels in breath and blood is still obscure, thus hindering the development of validated breath tests for diagnostic purposes. For perspective, multiplying the proposed population mean of approximately 1~$\mu$g/l~\cite{schwarzace} in end-tidal breath by the partition coefficient $\hen=340$~\cite{anderson2006} at body temperature appears to grossly underestimate observed (arterial) blood levels spreading around 1~mg/l~\cite{wangace,wigaeus1981,keller1984}. Furthermore, breath profiles of acetone (and other highly soluble volatile compounds such as 2-pentanone or methyl acetate) associated with moderate workload ergometer challenges of normal healthy volunteers drastically depart from the trend suggested by Equation~\eqref{eq:farhi}~\cite{setup,King2010GC}. In particular, the physiological meaning of these discrepancies has not been established in sufficient depth.  

With the background material of the previous paragraphs in mind, we view acetone as a paradigmatic example for the analysis of highly soluble, \emph{blood-borne} VOCs, even though it cannot cover the whole spectrum of different physico-chemical characteristics. The emphasis of this paper is on developing a mechanistic description of \textit{end-tidal} acetone behavior during different physiological states (e.g., rest, exercise, sleep and exposure scenarios). Such a quantitative approach will contribute to a better understanding regarding the relevance of observable breath concentrations of highly soluble trace gases with respect to the underlying endogenous situation and hence constitutes an indispensable prerequisite for guiding the interpretation of future breath test results. Moreover, it will allow for a standardized examination of the information content and predictive power of various breath sampling regimes proposed in the literature.
Specifically, our work also aims at complementing previous studies centered on single breath dynamics during resting conditions~\cite{tsu1988,johanson1991,anderson2003,anderson2006,kumagai2000}.

By adopting the usual compartmental approach~\cite{kumagai1995,clewell2001,moerk2006,reddybook} our formulation is simple in the sense that no detailed anatomical features of the respiratory tract must be taken into account. Although models of this type have been criticized for their underlying assumptions~\cite{hahn2003} (e.g., regarding the cyclic nature of breathing), they prove as valuable tools for capturing both short-term behavior as indicated above and phenomena that are characteristic for sampling scenarios extending over minutes or even hours. Consequently, while the physical derivation to be presented here is clearly driven by the well-established theory covering soluble gas exchange in a single exhalation framework, it extends these ideas to a macroscopic level, thus yielding a model that can serve as a template for studying the mid- to long-term kinetics of acetone and similar volatile organic compounds in breath and various parts of the human body

%
\section{Experimental basics}
\label{sect:experimental}
%
Here we shall briefly discuss the experimental background pertinent to our own phenomenological findings presented throughout the paper. In particular, these results were obtained with the necessary approvals by the Ethics Committee of the Innsbruck Medical
University. All volunteers gave written informed consent. \\

Breath acetone concentrations are assessed by means of a real-time setup designed for synchronized measurements of exhaled breath VOCs as well as a variety of respiratory and hemodynamic parameters, see Fig.~\ref{fig:setup}. Extensive details are given in~\cite{setup}.   

The breath-related part of the mentioned setup consists of a head mask spirometer system allowing for the standardized extraction of predefined exhalation segments which -- via a heated and gas tight Teflon transfer line -- are then directly drawn into a Proton-Transfer-Reaction mass spectrometer (PTR-MS, Ionicon Analytik GmbH, Innsbruck, Austria) for online analysis. This analytical technique has proven to be a sensitive method for quantification of volatile molecular species $M$ down to the ppb (parts per billion) range on the basis of ``soft'' chemical ionization within a drift chamber, i.e., by taking advantage of the proton transfer
\[\mathrm{H_3O}^+ + M \to M\mathrm{H}^+ + \mathrm{H_2O}\]
from primary hydronium precursor ions originating in an adjoint hollow cathode~\cite{lindinger1998,lindinger1998_2}. Note that this reaction scheme is selective to VOCs with proton affinities higher than water (166.5~kcal/mol), thereby precluding the ionization of the bulk composition exhaled air, N$_2$, O$_2$ and CO$_2$. Count rates of the resulting product ions $M\mathrm{H}^+$ or fragments thereof appearing at specified mass-to-charge ratios $m/z$ can subsequently be converted to absolute concentrations of the protonated compounds (see~\cite{schwarzfrag} for further details on the quantification of acetone as well as~\cite{setup} for the underlying PTR-MS settings used). The carbon dioxide concentration $\ccd$ of the gas sample is determined by a separate sensor (AirSense Model 400, Digital Control Systems, Portland, USA).

\begin{figure}[H]
\centering
\begin{tabular}{c}
\includegraphics[width=11cm]{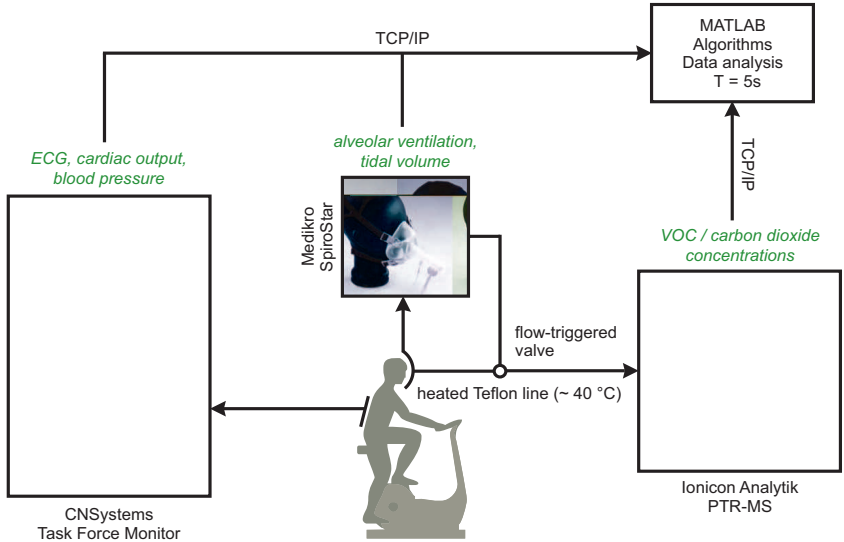}
\end{tabular}
\caption{Experimental setup used for obtaining VOC profiles and medical parameters~\cite{setup}. Items in italic correspond to measurable variables. The selective analysis of predefined breath segments is ensured by flow-triggered sample extraction.}\label{fig:setup}
\end{figure}

In addition to the breath concentration profiles of acetone, it will be of importance for us to have at hand a continuous estimate of the corresponding sample water vapor content $\cw$. As has been put forward in the literature, the water dimer $(\mathrm{H_3O}^+)\mathrm{H_2O}$ can be used for this purpose~\cite{warneke2001,ammann2006}. More specifically, the corresponding pseudo concentration signal at $m/z=37$ calculated according to Equation (1) in~\cite{schwarzfrag} using a standard reaction rate constant of $2.0 \times 10^{-9}$~cm$^3$/s yields a quantity roughly proportional to sample humidity. Slight variations due to fluctuations of the (unknown) amount of water clusters forming in the ion source are assumed to be negligible. Absolute quantification can be achieved by comparison with standards containing predefined humidity levels. Such standards with $\ccd$ and $\cw$ varying over the experimental physiological range of 2~--~8\% and 2~--~6\%, respectively, were prepared using a commercial gas mixing unit (Gaslab, Breitfuss Messtechnik GmbH, Harpstedt, Germany), resulting in a mean calibration factor of $2.1\times 10^{-4}$ and $R^2 \geq 0.98$ for all regressions. It has been argued in~\cite{keck2008} that the aforementioned pseudo concentration can drastically be affected by the carbon dioxide concentration, which, however, could not be confirmed with our PTR-MS settings. Although the computed water content is slightly overestimated with increasing $\ccd$, the sensitivity was found to stay within 10\% of the mean value given above. Nevertheless, we recognize that this approximate method for determining water vapor levels can only serve as a first surrogate for more exact hygrometer measurements. \\
Despite the fact that molecular oxygen is not protonated, the breath oxygen concentration $\coxy$ -- relative to an assumed steady state value of about 100~mmHg in end-tidal (alveolar) air during rest~\cite{lumbbook} -- within one single experiment can still be assessed by monitoring the parasitic precursor ion O$_2^+$ at $m/z=32$. This ion results from a small amount of sample gas entering the ion source with subsequent ionization of O$_2$ under electron impact~\cite{ohara2008}. For normalization purposes, the respective count rates are again converted to pseudo concentrations.
Table~\ref{table:measparams} summarizes the measured quantities relevant for this paper. In general, breath concentrations will always refer to end-tidal levels, except where explicitly noted. Moreover, a typical sampling interval of 5~s is assumed for each variable (corresponding to breath-by-breath extraction of end-tidal VOC levels at a normal breathing rate of 12~tides/min).     

\begin{table}[H]
\centering 
\caption{Summary of measured parameters together with some nominal values during rest and assuming ambient conditions. Breath concentrations refer to end-tidal levels.}\label{table:measparams}
\begin{tabular}{|lcc|}\hline
{\large\strut}Variable&Symbol&Nominal value (units)\\ \hline \hline 
{\large\strut}Cardiac output &$\qc$ & 6 (l/min)~\cite{mohrman2006}\\
{\large\strut}Alveolar ventilation &$\qalv$ & 5.2 (l/min)~\cite{westbook}\\
{\large\strut}Tidal volume &$\vt$ & 0.5 (l)~\cite{westbook}\\
{\large\strut}Acetone concentration &$\cmeas$ & 1 ($\mu$g/l)~\cite{schwarzace}\\
{\large\strut}CO$_2$ content &$\ccd$ & 5.6 (\%)~\cite{lumbbook}\\
{\large\strut}Water content &$\cw$ & 4.7 (\%)~\cite{hanna1986}\\
{\large\strut}O$_2$ content &$\coxy$ & 13.7 (\%)~\cite{lumbbook}\\
\hline
\end{tabular}
\end{table}

%
\section{Acetone modeling}
\label{sect:acemodel}
%
\subsection{Preliminaries and assumptions}
\label{sect:prelim}
Classical pulmonary inert gas elimination theory~\cite{farhi1967} postulates that uptake and removal of VOCs takes place exclusively in the alveolar region. While this is a reasonable assumption for low soluble substances, it has been shown by several authors that exhalation kinetics of VOCs with high affinity for blood and water such as acetone are heavily influenced by relatively quick absorption and release mechanisms occurring in the conductive airways (see, e.g.,~\cite{anderson2003} for a good overview of this topic). More specifically, due to their pronounced hydrophilic characteristics such compounds tend to interact with the water-like mucus membrane lining this part of the respiratory tree, thereby leading to pre- and post-alveolar gas exchange often referred to as wash-in/wash-out behavior.
The present model aims at taking into consideration two major aspects in this framework.

\subsubsection{Bronchial exchange}
\label{sect:broex}
It is now an accepted fact that the bronchial tree plays an important role in overall pulmonary gas exchange of highly (water) soluble trace gases, affecting both endogenous clearance as well as exogenous uptake. For perspective, Anderson et al.~\cite{anderson2006} inferred that while fresh air is being inhaled, it becomes enriched with acetone stored in the airway surface walls of the peripheral bronchial tract, thus leading to a decrease of the acetone pressure/tension gradient between gas phase and capillary blood in the alveolar space. This causes an effective reduction of the driving force for gas exchange in the alveoli and minimizes the unloading of the capillary acetone level. Correspondingly, during exhalation the aforementioned diffusion process is reversed, with a certain amount of acetone being stripped from the air stream and redepositing onto the previously depleted mucus layer. As a phenomenological consequence, exhaled breath concentrations of acetone and other highly water soluble substances tend to be diminished on their way up from the deeper respiratory tract to the airway opening, thereby decreasing overall elimination as compared to purely alveolar extraction. Similarly, exposition studies suggest a pre-alveolar absorption of exogenous acetone during inhalation and a post-alveolar revaporization during expiration, resulting in a lower systemic uptake compared to what would be expected if the exchange occurred completely in the alveoli~\cite{wigaeus1981,kumagai2000,thrall2003}.

From the above, quantitative assessments examining the relationships between the measured breath concentrations and the underlying alveolar levels are complex and need to take into account a variety of factors, such as airway temperature profiles and airway perfusion as well as breathing patterns~\cite{anderson2003,anderson2007}. \\

In accordance with previous modeling approaches, we consider a bronchial compartment separated into a gas phase and a mucus membrane, which is assumed to inherit the physical properties of water and acts as a reservoir for acetone~\cite{kumagai2000,moerk2006}. Part of the acetone dissolved in this layer is transferred to the bronchial circulation, whereby the major fraction of the associated venous drainage is postulated to join the pulmonary veins via the postcapillary anastomoses~\cite{lumbbook}. A study by Morris et al.~\cite{morris2008} on airway perfusion during moderate exercise in humans indicates that the fraction $\qbr \in [0,1)$ of cardiac output $\qc$ contributing to this part of bronchial perfusion will slightly decrease with increasing pulmonary blood flow. According to Figure~3 from their paper and assuming that $\qc^{\mathrm{rest}}=6$~l/min we can derive the heuristic linear model
\begin{equation}\label{eq:qbro}
\qbr(\qc):=\max\{0,\qbr^{\mathrm{rest}}(1-0.06(\qc-\qc^{\mathrm{rest}}))\}.
\end{equation}
The constant $\qbr^{\mathrm{rest}}$ will be estimated in Section~\ref{sect:simest}. As a rough upper bound we propose the initial guess $\qbr^{\mathrm{rest}}= 0.01$~\cite{lumbbook}. We stress the fact that the bronchial compartment just introduced has to be interpreted as an abstract control volume lumping together the decisive sites of airway gas exchange in one homogeneous functional unit. These locations can be expected to vary widely with the solubility of the VOC under scrutiny as well as with physiological boundary conditions~\cite{anderson2003}. 

\subsubsection{Temperature dependence}
\label{sect:temp}
There is strong experimental evidence that airway temperature constitutes a major determinant for the pulmonary exchange of highly soluble VOCs, cf.~\cite{jones1982}. In particular, changes in airway temperature can be expected to affect the solubility of acetone and similar compounds in the mucus surface of the respiratory tree. It will hence be important to specify a tentative relationship capturing such influences on the observable breath levels. As will be rationalized below, this may be achieved by taking into account the absolute humidity of the extracted breath samples.

Passing through the conditioning regions of the upper airways, inhaled air is warmed to a mean body core temperature of approximately 37$\phantom{}^{\circ}$C and fully saturated with water vapor, thus leading to an absolute humidity of alveolar air of about 6.2\% at ambient pressure. During exhalation, depending on the axial temperature gradient between the lower respiratory tract and the airway opening, a certain amount of water vapor will condense out and reduce the water content $\cw$ in the exhalate according to the saturation water vapor pressure $P_{\mathrm{water}}$ (in mbar) determined by local airway temperature $T$ (in $\phantom{}^{\circ}$C)~\cite{mcfadden1985,hanna1986a}. 
The relationship between these two quantities can be approximated by the well-known Magnus formula~\cite{sonntag} 
\begin{equation}\label{eq:magnus}P_{\mathrm{water}}(T)=6.112\,\exp{\lk\frac{17.62\,T}{243.12+T}\rk},\end{equation}
valid for a temperature range $-45^{\circ}\textrm{C} \leq T \leq 60^{\circ}\textrm{C}$. 
For normal physiological values of $T$, the resulting pressure is sufficiently small to treat water vapor as an ideal gas~\cite{raobook} and hence by applying Dalton's law we conclude that absolute humidity $\cw$ (in $\%$) of the exhalate varies according to
\begin{equation}\cw(T)=100\,\frac{P_{\mathrm{water}}(T)}{P_{\mathrm{ambient}}},\end{equation}
where $P_{\mathrm{ambient}}$ is the ambient pressure.
Inverting the above formula, the mini\-mum airway temperature $T_{\mathrm{min}}=T_{\mathrm{min}}(\cw)$ during exhalation becomes a function of measured water content in exhaled breath. From this, a mean airway and mucus temperature characterizing the homogeneous bronchial compartment of the previous section will be defined as
\begin{equation}\label{eq:temp}
\bar{T}(\cw):=\frac{T_{\mathrm{min}}(\cw)+37}{2},
\end{equation}  
corresponding to a hypothesized linear increase of temperature along the airways. Note that this assumption is somehow arbitrary in the sense that the characteristic temperature of the airways should be matched to the primary (time- and solubility-dependent) location of airway gas exchange as mentioned above. Equation~\eqref{eq:temp} thus should only be seen as a simple ad hoc compromise incorporating this variability.\\

The decrease of acetone solubility in the mucosa -- expressed as the water:air partition coefficient $\lmuca$ -- with increasing temperature can be described in the ambient temperature range by a van't Hoff-type equation~\cite{staudinger2001}
\begin{equation}\label{eq:hentemp}\log_{10}\lmuca(T)=-A+\frac{B}{T+273.15},\end{equation}
where $A=3.742$ and $B=1965$~Kelvin are proportional to the entropy and enthalpy of volatilization, respectively. Hence, in a hypothetical situation where the absolute sample humidity at the mouth is 4.7\% (corresponding to a temperature of $T \approx 32^{\circ}\mathrm{C}$ and ambient pressure at sea level, cf.~\cite{mcfadden1985,hanna1986}), local solubility of acetone in the mucus layer increases from $\lmuca(37^{\circ}\mathrm{C}) = 392$ in the lower respiratory tract (cf.~\cite{kumagai1995}) to $\lmuca(32^{\circ}\mathrm{C}) = 498$ at the mouth, thereby predicting a drastic reduction of air stream acetone concentrations along the airways. The above formulations allow one to assess this reduction by taking into account sample water vapor as a meta parameter. This meta parameter reflects various influential factors on the mucus solubility $\lmuca$ which would otherwise be intricate to handle due to a lack of information, such as local airway perfusion, breathing patterns, mucosal hydration and thermoregulatory events which in turn will affect axial temperature profiles. In particular, $\lmuca$ for the entire bronchial compartment will be estimated via the mean airway temperature $\bar{T}$ in Equation~\eqref{eq:temp} as
\begin{equation}\label{eq:meanhen}
\lmuca(\bar{T})=\lmuca(\bar{T}(\cw)).
\end{equation} 
The strong coupling between sample humidity and exhaled breath concentrations predicted by the two relationships~\eqref{eq:magnus} and~\eqref{eq:hentemp} is expected to be a common factor for all highly water soluble VOCs. In the framework of breath alcohol measurements Lindberg et al.~\cite{lindberg2007} indeed showed a positive correlation between these two quantities along the course of exhalation, which can also be observed in the case of acetone, cf.~Fig.~\ref{fig:lin}.

\begin{figure}
\centering
\begin{tabular}{c}
\includegraphics[height=12cm]{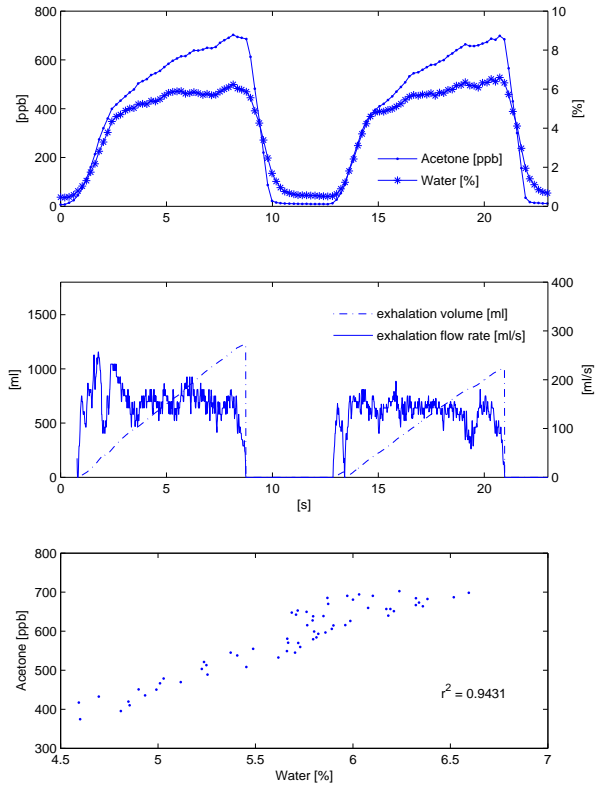}
\end{tabular}
\caption{Correlation between breath acetone concentrations and breath water content during two consecutive exhalations of a normal healthy volunteer at an approximately constant rate of 150~ml/s. Water-acetone pairs in the third panel correspond to the two linear \emph{Phase~3} segments as described in~\cite{anderson2006}. The high sampling frequency in this example is achieved by collecting breath over the entire breath cycle (as opposed to the selective end-tidal extraction regime described in Section~\ref{sect:experimental}) as well as by limiting PTR-MS detection to only three mass-to-charge ratios (with corresponding dwell times given in brackets): $m/z=21$ (5~ms), $m/z=37$ (2~ms) and $m/z=59$ (10~ms).}\label{fig:lin}
\end{figure}

Variations of the acetone blood:air partition coefficient $\hen=340$~\cite{anderson2006,crofford1977} -- dominating alveolar gas exchange -- in response to changes in mixed venous blood temperature, e.g., due to exercise, are ignored as such changes are necessarily small~\cite{brundin1975}. Hence, $\hen$ will always refer to 37$^{\circ}$C. Similarly, the partition coefficient between mucosa and blood is treated as a constant defined by
\begin{equation}\label{eq:lmucb}\lmucb:=\lmuca(37^{\circ}\mathrm{C})/\hen, \end{equation} resulting in a value of 1.15. Note, that if the airway temperature is below 37$\phantom{}^{\circ}$C we always have that
\begin{equation}\label{eq:thermdis}\lmuca/\lmucb \geq \hen,\end{equation}
as $\lmuca$ is monotonically decreasing with increasing temperature, see Equation~\eqref{eq:hentemp}. 

\subsubsection{Bronchio-alveolar interactions}
\label{sect:broalvinter}
In a series of modeling studies~\cite{tsu1988,anderson2003}, the location of gas exchange has been demonstrated to shift between bronchial and alveolar regions depending on the solubility of the compound under investigation. During tidal breathing exchange for substances with blood:air partition coefficient $\hen \leq 10$  takes place almost exclusively in the alveoli, while it appears to be strictly limited to the bronchial tract in the case of $\hen \geq 100$. Transport for VOCs lying within these two extremes distributes between both spaces. Likewise, for fixed $\hen$, the location of gas exchange is expected to vary with 
breathing patterns. As has been concluded by Anderson et al.~\cite{anderson2006}, airway contribution to overall pulmonary exchange of endogenous acetone is about 96\% during tidal breathing, but only 73\% when inhaling to total lung capacity. The rationale for this reduction is that while more proximal parts of the mucosa lining are being depleted earlier in the course of inhalation by losing acetone to the inhalate, saturation of the air stream with acetone is continuously shifted towards the alveolar region. Furthermore, it can be argued that the magnitude of this shift increases with volumetric flow during inhalation as equilibration of fresh air with the mucus layer in regions with high flow rates might not be completed. From a reversed viewpoint this would be consistent with the observation made in~\cite{anderson2006} that end-exhaled acetone partial pressures increase with exhaled flow rate. \\

Such smooth transitions in the location of gas exchange can be incorporated into the model by including a diffusion process describing the interaction between bronchial and alveolar compartment (cf.~Fig.~\ref{fig:model_struct}), which is similar to the strategy found in the theory of stratified or series inhomogeneities developed by Scheid et al.~\cite{scheid1981}. According to the approach presented there, series inhomogeneities stem from the fact that while gas flow in the upper parts of the respiratory tree is primarily dominated by convection, axial diffusion becomes the decisive factor in the terminal airspaces. This can also be thought of as an incomplete mixing of tidal volume with the functional residual capacity, thus leading to the formation of an effective gas diffusion barrier between the proximal and distal parts of the alveolar space. Scheid et al. quantify this effect by means of a (substance-specific) stratified conductance parameter $D$, taking values in the interval $[0,\infty)$. As $D$ approaches zero, the retention (defined as the ratio between steady state partial pressures in arterial and mixed venous blood) of inert gases eliminated from the blood increases, while excretion (the ratio between partial pressures in mixed expired air and mixed venous blood) decreases. Hence, small values of $D$ correspond to a reduced overall gas exchange efficiency of the lungs. It will be shown that reinterpreting the concept of stratified inhomogeneity and the stratified conductance parameter $D$ in the framework of soluble gas exchange allows for a proper description of the bronchio-alveolar interactions discussed above. The particular role of $D$ within the present framework will be clarified in Section~\ref{sect:steadystates}.\\

The alveolar region itself is represented by one single homogeneous alveolar unit, thereby neglecting ventilation-perfusion inequality throughout the lung.
In the case of VOCs with high $\hen$ this constitutes an acceptable simplification, since the classical Farhi equation predicts a minimal influence of local ventilation-perfusion ratios on the corresponding alveolar concentrations. Uptake and elimination of VOCs to and from the bronchio-alveolar tract during inhalation and exhalation is governed by the alveolar ventilation $\qalv$ (defined as the gas volume per time unit filling the alveoli and the exchanging bronchial tubes).

\subsubsection{Body compartments}
\label{sect:bodycomp}
The systemic part of the model has been adapted from previous models~\cite{kumagai1995,moerk2006} and consists of two functional units: a liver compartment, where acetone is endogenously produced and metabolized, as well as a tissue compartment representing an effective storage volume. The latter basically lumps together tissue groups with similar blood:tissue partition coefficient $\ltis \approx 1.38$, such as richly perfused tissue, muscles and skin~\cite{anderson2006,moerk2006}. Due to the low fractional perfusion of adipose tissue and its low affinity for acetone, an extra fat compartment is not considered. The fractional blood flow $\qliv \in (0,1)$ to the liver is assumed to be related to total cardiac output by
\begin{equation}\qliv(\qc):=0.034+1.145\exp{(-1.387\qc/\qc^{\mathrm{rest}})},\end{equation}
obtained by exponential fitting of the data given in~\cite{moerk2006}.
The rate $\dot{M}$ of acetone metabolism is assumed to obey the Michaelis-Menten (saturation) kinetics
\begin{equation}\label{eq:metmenten}\dot{M}=\frac{\vmax\cliv\lliv}{\km+\cliv\lliv},\end{equation}
with $\vmax,\,\km>0,$ or the linear kinetics
\begin{equation}\label{eq:ratemet}\dot{M}=\klin\bw^{0.75}\cliv\lliv=:\met\cliv\lliv,\end{equation}
where ``$\bw$'' denotes the body weight in~kg. The rate constant $\klin$ can be obtained by linearization of~\eqref{eq:metmenten}, i.e., 
\begin{equation}\klin:=\vmax/\km \approx 0.0037\;\mathrm{l}/\mathrm{kg}^{0.75}/\mathrm{min},\end{equation}
according to the values given in~\cite{kumagai1995}.

\begin{remark}
Taking into account nominal mixed venous blood concentrations of approximately 1~mg/l for healthy volunteers as well as an apparent Michaelis-Menten constant $\km = 84$~mg/l~\cite{kumagai1995}, linear kinetics will be sufficient for describing most situations encountered in practice. Typical exceptions include, e.g., severe diabetic ketoacidosis or starvation ketosis, where plasma concentrations up to $500$~mg/l have been reported~\cite{owen1982,reichard1979} and hence metabolism can be expected to reach saturation.
\end{remark}
Other ways of acetone clearance such as excretion via the renal system are neglected, cf.~\cite{wigaeus1981}.
Endogenous synthesis of acetone in the liver is assumed to occur at some rate $\pr>0$ depending on current lipolysis~\cite{kalapos2003}. In particular, fat catabolism is considered a long-term mechanism compared to the other dynamics of the system, so that $\pr$ can in fact be assumed constant during the course of experiments presented here (less than 2 hours).

\subsection{Model equations and a priori analysis}
\subsubsection{Derivation}
In order to capture the gas exchange and tissue distribution mechanisms presented above, the model consists of four different compartments.
A sketch of the model structure is given in Fig.~\ref{fig:model_struct} and will be detailed in following. 

Model equations are derived by taking into account standard conservation of mass laws for the individual compartments, see Appendix~\ref{sect:compmod}.
Local diffusion equilibria are assumed to hold at the air-tissue, tissue-blood and air-blood interfaces, the ratio of the corresponding concentrations being described by the appropriate partition coefficients, e.g., $\hen$. Unlike for low blood soluble compounds, the amount of highly soluble gas dissolved in local blood volume of perfused compartments cannot generally be neglected, as it might significantly increase the corresponding capacities. This is particularly true for the airspace compartments. Since reliable data for some local blood volumes could not be found, in order not to overload the model with too many hypothetical parameters, we will use the effective compartment volumes $\vbro:=V_{\mathrm{bro}}+V_{\mathrm{muc}}\lmuca$, 
$\valv:=V_{\mathrm{A}}+V_{\mathrm{c'}}\hen$, $\vliv:=V_{\mathrm{liv}}+V_{\mathrm{liv,b}}\lliv$ as well as $\vtis:=V_{\mathrm{tis}}$ and neglect blood volumes for the mucosal and tissue compartment, see also Appendix~\ref{sect:compmod}.

\begin{figure}[H]
\centering
\begin{picture}(7,10)

\put(0.7,9.8){$C_\mathrm{I}$}
\put(0.8,9.7){\vector(0,-1){0.95}}
\put(1.4,9.8){$C_\mathrm{bro}$}
\put(1.5,8.75){\vector(0,1){0.95}}
\put(1,9.2){$\qalv$}

\put(3.9,8.15){\vector(-1,0){0.75}}
\put(3.9,8.15){\line(0,-1){6.8}}

\put(4.8,8.1){\parbox{1cm}{bronchial\\compartment}}
\put(4,8){\rotatebox{270}{$q_\mathrm{bro} \dot{Q}_\mathrm{c}$}}
\put(0.5,8){\fbox{\rule[-3ex]{0pt}{8ex}\parbox{1cm}{\vspace{-1mm}\centering $C_\mathrm{bro}$\\[2mm]$V_\mathrm{bro}$} \quad \parbox{1cm}{\vspace{-1mm}\centering $C_\mathrm{muc}$\\[2mm]$V_\mathrm{muc}$} }}
\multiput(1.8,8.7)(0,-0.2){6}{\line(0,-1){0.1}}

\put(1.3,7.05){$D$}
\put(1.15,7.5){\vector(0,-1){0.75}}
\put(1.15,6.75){\vector(0,1){0.75}}
\put(2.5,7.5){\line(0,-1){0.75}}
\put(2.5,7.125){\vector(1,0){1.4}}

\put(4.8,6.1){\parbox{1cm}{alveolar\\compartment}}
\put(4,6.5){\rotatebox{270}{$(1-q_\mathrm{bro}) \dot{Q}_\mathrm{c}$}}
\put(0.5,6){\fbox{\rule[-3ex]{0pt}{8ex}\parbox{1cm}{\vspace{-1mm}\centering $C_\mathrm{A}$\\[2mm]$V_\mathrm{A}$} \quad \parbox{1cm}{\vspace{-1mm}\centering $C_\mathrm{c'}$\\[2mm]$V_\mathrm{c'}$} }}
\multiput(1.8,6.7)(0,-0.2){6}{\line(0,-1){0.1}}

\put(4.8,3.65){\parbox{1cm}{liver\\compartment}}
\put(4,4.025){\rotatebox{270}{$(1-q_\mathrm{liv})(1-q_\mathrm{bro}) \dot{Q}_\mathrm{c}$}}
\put(0.5,3.75){\fbox{\rule[-9ex]{0pt}{16ex}\parbox{1cm}{\vspace{1mm}\centering  $C_\mathrm{liv,b}$\\[2mm]$V_\mathrm{liv,b}$\\[4mm] $C_\mathrm{liv}$\\[2mm]$V_\mathrm{liv}$} }}
\multiput(0.6,3.75)(0.2,0){6}{\line(1,0){0.1}}

\put(2,4.5){$q_\mathrm{liv} (1-q_\mathrm{bro}) \dot{Q}_\mathrm{c}$}
\put(3.9,4.35){\vector(-1,0){2.1}}

\put(2.85,3.3){$k_\mathrm{met}$}
\put(1.8,3.4){\vector(1,0){0.95}}
\put(2.85,2.7){$k_\mathrm{pr}$}
\put(2.75,2.8){\vector(-1,0){0.95}}

\put(4.8,1.25){\parbox{1cm}{tissue\\compartment}}
\put(0.5,1.25){\fbox{\rule[-3ex]{0pt}{8ex}\parbox{1cm}{\vspace{-1mm}\centering $C_\mathrm{tis}$\\[2mm]$V_\mathrm{tis}$} }}

\put(3.9,1.35){\vector(-1,0){2.1}}
\put(0.5,1.35){\vector(-1,0){0.25}}
\put(0.25,1.35){\line(0,1){3.9}}
\put(0.5,4.35){\vector(-1,0){0.25}}
\put(0.25,5.25){\line(1,0){2.25}}
\put(2.5,5.25){\vector(0,1){0.25}}

\end{picture}
\caption{Sketch of the model structure. The body is divided into four distinct functional units: bronchial/mucosal compartment (gas exchange), alveolar/end-capillary compartment (gas exchange), liver (metabolism and production) and tissue (storage). Dashed boundaries indicate a diffusion equilibrium. The conductance parameter $D$ has units of volume divided by time and quantifies an effective diffusion barrier between the bronchial and the alveolar tract, cf.~Section~\ref{sect:steadystates}.}\label{fig:model_struct}
\end{figure}
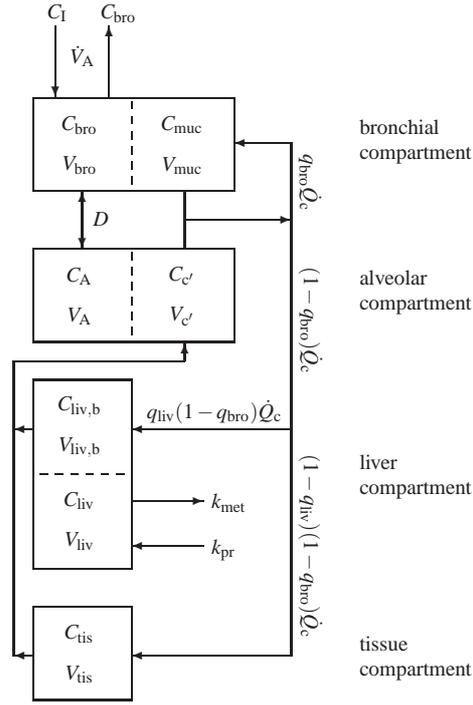

According to Fig.~\ref{fig:model_struct} as well as by taking into account the discussion of the previous subsections, for the bronchial compartment we find that
\begin{equation}\label{eq:bro}
\frac{\di\cbr}{\di t}\vbro=\qalv(\cinh-\cbr)+D(\calv-\cbr)\\+\qbr\qc\lk\cart-\frac{\lmuca}{\lmucb}\cbr\rk,
\end{equation}
with $\cinh$ denoting the inhaled (ambient) gas concentration, while the mass balance equations for the alveolar, liver and tissue compartment read
\begin{equation}\label{eq:alv}
\frac{\di\calv}{\di t}\valv=D(\cbr-\calv)\\+(1-\qbr)\qc\big(\cven -\hen\calv\big),
\end{equation} 
and
\begin{equation}\label{eq:liv}
\frac{\di\cliv}{\di t}\vliv=\pr-\met\lliv\cliv\\+\qliv(1-\qbr)\qc\big(\cart-\lliv\cliv\big),
\end{equation}
and
\begin{equation}\label{eq:tis}
\frac{\di\ctis}{\di t}\vtis=(1-\qliv)(1-\qbr)\qc\big(\cart-\ltis\ctis\big),
\end{equation}
respectively. Here,
\begin{equation}\label{eq:cven}\cven:=\qliv\lliv\cliv+(1-\qliv)\ltis\ctis\end{equation}
and
\begin{equation}\label{eq:cart}\cart:=(1-\qbr)\hen\calv+\qbr\lmuca\cbr/\lmucb\end{equation}
are the associated concentrations in mixed venous and arterial blood, respectively.
Moreover, we state that measured (end-tidal) breath concentrations equal bronchial levels, i.e.,
\begin{equation}\label{eq:meas}\cmeas=\cbr.\end{equation}
The decoupled case $D=\qbr=0$ will be excluded further on in this paper as it lacks physiological relevance.\\

Some fundamental model properties are discussed in Appendix~\ref{sect:apriori}. In particular, the components of the state variable $\cb:=(\cbr,\calv,\cliv,\ctis)^T$ remain non-negative, bounded and will approach a globally asymptotically stable equilibrium $\cb^e(\ub)$ once the measurable external inputs $\ub:=(\qalv,\qc,\vt,\lmuca(\cw),\cinh)$ affecting the system are fixed. This corresponds, e.g., to the situation encountered during rest or constant workload.

\subsubsection{Steady state relationships and interpretation of the stratified conductance parameter $D$}\label{sect:steadystates}
Assume that the system is in steady state and we know the associated end-tidal breath concentration $\cmeas=\cbr$.   
Furthermore, we define the bronchial and alveolar ventilation-perfusion ratio as
\[r_{\mathrm{bro}}:=\frac{\qalv}{\qbr\qc}\quad \textrm{and}\quad r_{\mathrm{A}}:=\frac{\qalv}{(1-\qbr)\qc},\]
respectively.
If $D=0$ we deduce that
\begin{multline}\label{eq:Dzero}\cmeas=\cbr=\\\frac{r_{\mathrm{bro}}\cinh+(1-\qbr)\hen\calv}{(1-\qbr)\frac{\lmuca}{\lmucb}+r_{\mathrm{bro}}}= \frac{r_{\mathrm{bro}}\cinh+(1-\qbr)\cven}{(1-\qbr)\frac{\lmuca}{\lmucb}+r_{\mathrm{bro}}}=\frac{r_{\mathrm{bro}}\cinh+\cart}{\frac{\lmuca}{\lmucb}+r_{\mathrm{bro}}},\end{multline}
corresponding to purely bronchial gas exchange. On the other hand, for $D\to \infty$ it can be shown by simple algebra that
\begin{multline}\label{eq:Dinfty}\cmeas=\cbr=\\\calv= \frac{r_{\mathrm{A}}\cinh +\cven}{\qbr\frac{\lmuca}{\lmucb}+(1-\qbr)\hen+r_{\mathrm{A}}}=\frac{\cart}{\qbr\frac{\lmuca}{\lmucb}+(1-\qbr)\hen}.\end{multline}
In the following, let $\cinh=0$. Note that then in both cases the physiological boundary condition $\cven \geq \cart$ is respected. Substituting $\qbr=0$ into~\eqref{eq:Dinfty} yields the well-known Farhi equations describing purely alveolar gas exchange. \\
Consequently, $D$ defines the location of gas exchange in accordance with Section~\ref{sect:broalvinter}, while $\qbr$ determines to what extent the bronchial compartment acts as an inert tube. 
In this sense, Equations~\eqref{eq:bro}--\eqref{eq:tis} define a generalized description of gas exchange including several known models as special cases. This hierarchy is summarized in Fig.~\ref{fig:model_hier}.

\begin{figure}[H]
\centering
\begin{picture}(6.9,5)
\put(2.3,4.3){\fbox{\quad\rule[-1.55ex]{0pt}{4.2ex} model \quad}}
\put(2.5,3.9){\vector(-1,-1){0.8}}
\put(3.8,3.9){\vector(1,-1){0.8}}
\put(0.8,3.5){$\qbr=0$}
\put(4.5,3.5){$\qbr>0$}

\put(0.1,2.5){\fbox{\parbox{2.6cm}{\centering stratified\\ inhomogeneity}}}
\put(3.1,2.5){\fbox{\parbox{3cm}{\centering location dependent\\ gas exchange}}}
\put(1.5,2){\vector(0,-1){0.9}}
\put(3.9,2){\vector(0,-1){0.9}}
\put(5.6,2){\vector(0,-1){0.9}}
\put(1.7,1.5){$D \to \infty$}
\put(4.1,1.5){$D \to 0$}
\put(5.8,1.5){$D \to \infty$}

\put(0.1,0.5){\fbox{\parbox{2.6cm}{\centering Farhi inert gas\\ elimination}}}
\put(3.1,0.5){\fbox{\rule[-1.55ex]{0pt}{4.2ex}\parbox{1.3cm}{\centering bronchial}}}
\put(4.8,0.5){\fbox{\rule[-1.55ex]{0pt}{4.2ex}\parbox{1.3cm}{\centering alveolar}}}

\end{picture}
\caption{Equations~\eqref{eq:bro}--\eqref{eq:tis} viewed as generalized model including several gas exchange mechanisms as special cases ($\cinh=0$).}\label{fig:model_hier}
\end{figure}
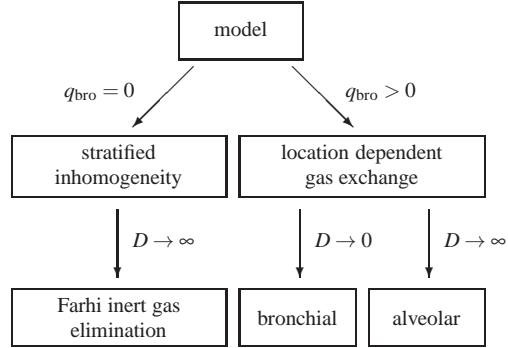

For perspective -- as has been rationalized in the case of acetone in Section~\ref{sect:broalvinter} and is likely to be a common characteristic for all highly water soluble VOCs -- the stratified conductance parameter $D$ will be close to zero during rest and is expected to increase with tidal volume and/or flow rate~\cite{tsu1991,anderson2003}. We propose to model this dependency as
\begin{equation}\label{eq:D}
D:=D^{\mathrm{rest}}+\dke \max\{0,\vt-\vt^{\mathrm{rest}}\}+\dkz  \max\{0,\qalv-\qalv^{\mathrm{rest}}\},\;k_{\mathrm{diff},j} \geq 0,
\end{equation}
which will further be justified in Section~\ref{sect:ergo}.

From a practical point of view, we stress the fact that if $D$ is close to zero, calculating blood levels from $\cmeas$ by using Equation~\eqref{eq:Dzero} or taking advantage of this expression for normalization and/or correction purposes is critical due to the possibly large influence of the term $r_{\mathrm{bro}}$ as well as due to the high degree of uncertainty with respect to $\qbr$ and $\lmuca$. 
Particularly, from the aforementioned facts this means that measured breath concentrations of acetone (and generally highly water soluble substances) determined during resting conditions and free breathing can be rather misleading indicators for endogenous levels, even if sampling occurs under well defined standard conditions (for instance -- as is common practice -- using CO$_2$- and/or flow-controlled extraction from the end-tidal exhalation segment). This has first been recognized in the context of breath alcohol measurements revealing experimentally obtained blood-breath concentration ratios of ethanol during tidal breathing that are unexpectedly high compared with in vitro partition coefficients~\cite{jones1983}. An elegant approach proposed to circumvent this problem is isothermal rebreathing~\cite{jones1983,ohlsson1990}, which aims at creating conditions where (cf.~Equation~\eqref{eq:Dzero})
\begin{equation}\cmeas=\cinh=\calv=\cbr = \frac{r_{\mathrm{bro}}\cbr+(1-\qbr)\cven}{(1-\qbr)\hen+r_{\mathrm{bro}}}=\frac{\cven}{\hen}=\frac{\cart}{\hen},\end{equation}
from which the mixed venous blood concentration can easily be determined by multiplying the measured rebreathing concentration with the anticipated blood:air partition coefficient at body temperature. A quantitative evaluation of this experimental technique by means of the above model can be found in the technical report~\cite{King2010b}.

%
\section{Model validation and estimation}
\label{sect:simest}
%
\subsection{A priori identifiability}\label{sect:identif}
One of the purposes of our tentative model is to provide a basis for estimating (unknown) compartment concentrations $\cb$ as well as certain acetone-specific parameters $p_j \in \{\pr,\met,\vmax,\km,D,\qbr^{\mathrm{rest}}\}$ from the knowledge of measured breath and blood concentrations $y$. Often over-looked, a necessary requirement in this framework is the a priori (or structural) identifiability/observability of the model, which basically checks whether in an ideal context of an error-free model and continuous, noise-free measurements there exist functions $\ub$ (or, in other words, conductible experiments) such that the associated output $y$ (the accessible data) carries enough information to enable an unambiguous determination of all unknown states and parameters. Particularly, this avoids an inherent over-parameterization of the model. As the time evolution of the system~\eqref{eq:bro}--\eqref{eq:tis} for a given $\ub$ is fixed once the initial conditions $\cb_0$ at the start of the experiment are known, the analysis of a priori identifiability/observability hence amounts to studying (local) injectivity of $y$ with respect to $\cb_0$ and the parameters $p_j$ under scrutiny.
Evidently, if such a property does not hold then any attempt to reliably estimate these quantities from $y$ is doomed to failure from the start, as two entirely different parameter combinations can yield exactly the same data.
In the present context, \emph{generic} a priori identifiability/observability for $\cb_0$ and any selection of parameters $p_j$ as above was confirmed by interpreting the latter as additional states with time derivative zero and subjecting the augmented system to the Hermann-Krener rank criterion~\cite{krener} (see Appendix~\ref{sect:krener}). In particular, the sufficient condition in~\eqref{eq:defJ} was fulfilled for both $y=\cbr$ and $y=\cart$ as given in Equations~\eqref{eq:meas} and~\eqref{eq:cart}, respectively.  

\subsection{Simulation of exposure data and model calibration}\label{sect:wigaeus}
Generally speaking, in vivo data on acetone \emph{dynamics} in the human organism are very limited. Indeed, experimental efforts to date have centered on quantifying bioaccumulation/biotransformation as well as body burden within occupational exposure settings. The study by Wigaeus et al.~\cite{wigaeus1981} represents the most extensive research in this context (including breath as well as simultaneous blood measurements), thus rendering it as a convenient benchmark for confirming the appropriateness of models involving descriptions of systemic acetone distribution, cf.~\cite{kumagai1995,moerk2006}.

In the following we will calibrate the unknown kinetic rate constants, the fractional bronchial blood flow, as well as the equilibrium tissue levels $\cb_0$ at rest, by subjecting the proposed model to the \emph{Series~1} exposure scenario published in~\cite{wigaeus1981}. Briefly, this data set comprises acetone concentration profiles in exhaled breath, arterial and (presumably peripheral) venous blood of eight normal male volunteers, who were exposed to an atmosphere containing $1.3$~mg/l of acetone over a period of two hours. Particularly, all measurements correspond to resting conditions. 

Since the major goal here is to extract approximate nominal values for the above-mentioned physiological parameters rather than individual estimates, the model will be fitted to the \emph{pooled} data of all eight test subjects. For this purpose, we assume that the exposure starts after ten minutes of quiet tidal breathing, i.e., the inhaled concentration is given by the scaled indicator function $\cinh(t)=1.3\,\chi_{[10,130]}(t)$ and set alveolar ventilation and cardiac output to constant resting values $\qalv^{\mathrm{rest}}=6 $~l/min and $\qc^{\mathrm{rest}}=5.8$~l/min, respectively~\cite{moerk2006,kumagai1995}. Tissue volumes and partition coefficients are as in Table~\ref{table:param} for a male of height 180~cm and weight 70~kg. 
Since the nominal acetone concentration in hepatic venous blood is a priori unknown, acetone metabolism is assumed to follow a Michaelis-Menten kinetics with a fixed apparent Michaelis constant $\km=84$~mg/l (cf.~Equation~\eqref{eq:metmenten}).
Our aim is to determine the parameter vector $\tb=(\vmax,\pr,D^{\mathrm{rest}},\qbr^{\mathrm{rest}})$ as well as the nominal endogenous steady state levels $\cb_0$ by solving the ordinary least squares problem
\begin{equation}\label{eq:OLS}\underset{{\cb_0,\tb}}{\mathrm{arg\,min}} \sum_{i=0}^N \big(C_{\mathrm{a},i}-\cart(t_i)\big)^2,\quad \mathrm{s.t.} \left\{\begin{array}{ll}
\gb(\ub_0,\cb_0,\tb)=0 & \textrm{(steady state)}\\
\cb_0,\tb \geq 0 & \textrm{(positivity)}\\
\qbr^{\mathrm{rest}} \leq 0.05 & \textrm{(normalization)}\\
\cart(0)=1\;\textrm{mg/l} & \textrm{(endog. arterial level)} \end{array} \right.\end{equation} 
Here, $\gb$ is the right-hand side of the ODE system~\eqref{eq:bro}--\eqref{eq:tis}, $C_{\mathrm{a},i}$ is the measured arterial blood concentration at time instant $t_i$, whereas the predicted arterial concentration $\cart$ is defined as in Equation~\eqref{eq:cart}. The endogenous level $\cart(0)=1\;\textrm{mg/l}$ was chosen in accordance with the population mean values given in~\cite{kalapos2003,wigaeus1981}. 

The above minimization problem was solved by implementing a multiple shooting routine~\cite{bock1987} in \textit{Matlab}. This iterative method can be seen as a generalization of the standard Gauss-Newton algorithm for solving constrained ordinary least squares problems, treating~\eqref{eq:OLS} in the framework of multipoint boundary value problems (see also~\cite{stoerbook} for an early illustration). Further details regarding the general scope of multiple shooting as well as its superior stability compared with classical solution schemes can be found in~\cite{bock1981,bock1987,pfeifer2007,voss2004}. For a variety of applications and modifications proposed for covering PDEs and delay differential equations the interested reader is referred to~\cite{muller2004,horbelt2002}.

Derivatives of $\cart$ with respect to $(\cb_0,\tb)$ were computed by simultaneously solving the associated variational equations~\cite{wannerbook}. The minimization procedure was repeated several times with randomly assigned starting values in the interval $(0,1)$. Convergence was assumed to be achieved when the maximum componentwise relative change between two successive parameter refinements was less than $0.1\%$, resulting in the same estimates for all trials considered, cf.~Table~\ref{table:fitwigaeus}. Figure~\ref{fig:wigaeus} shows that the calibrated model can faithfully reproduce the basic features of the observed data.

\begin{remark}\label{rem:ident}
The practical identifiability of the estimates in Table~\ref{table:fitwigaeus} was examined by calculating the rank of the extended Jacobian $J:=\begin{pmatrix}S^T & Z^T\end{pmatrix}$, where $S$ is the sensitivity function matrix having rows
\begin{equation}
S_{i,-}:=\begin{pmatrix}\frac{\partial C_{\mathrm{a}}(t_{i-1})}{\partial \tb} & & \frac{\partial C_{\mathrm{a}}(t_{i-1})}{\partial \cb_0}\end{pmatrix},
\end{equation}
and $Z$ denotes the Jacobian associated with the equality constraints in~\eqref{eq:OLS}.
More specifically, we adopted the standard \emph{numerical} rank criterion
\begin{equation}
\mathrm{rank}\;J=\max\{k;\;\sigma_k >\varepsilon \norm{J}_{\infty}\},
\end{equation}
where $\sigma_1\geq \sigma_2 \geq \ldots \geq 0$ are the singular values of $J$ and $\varepsilon = 10^{-8}$ reflects the maximum relative error of the calculated sensitivities~\cite{golubbook}. Accordingly, we find that $J$ has full rank, suggesting that all estimated quantities are practically identifiable~\cite{cobelli1980,jac1990}.
\end{remark}

\begin{figure}[H]
\centering
\begin{tabular}{c}
\hspace*{-0.8cm}
\includegraphics[height=9cm]{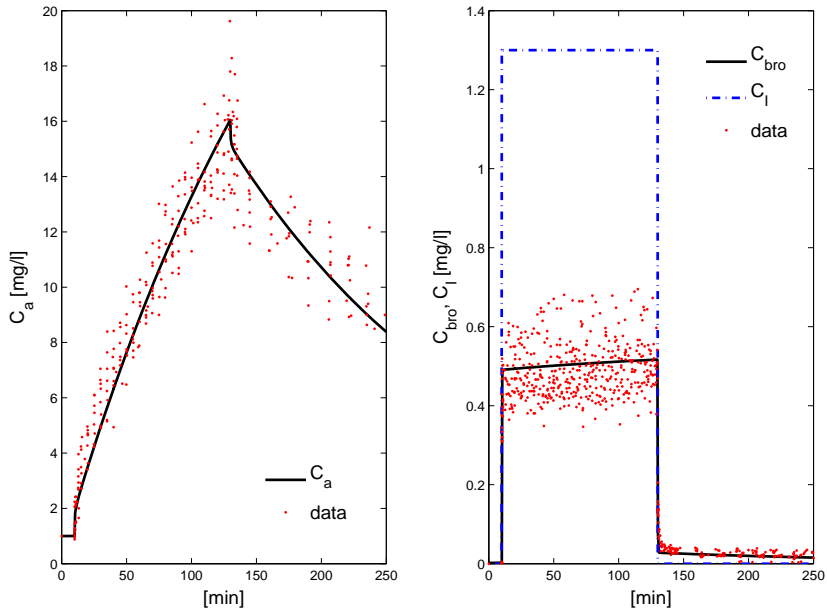}
\end{tabular}
\caption{Model fitted to the \emph{Series~1} exposure data published by Wigaeus et al.~\cite{wigaeus1981}. Data correspond to \emph{pooled} observations from eight normal healthy volunteers. First panel: observed versus predicted arterial concentrations, cf.~Equation~\eqref{eq:cart}. Second panel: observed versus predicted breath concentrations, cf.~Equation~\eqref{eq:meas}.}\label{fig:wigaeus}
\end{figure}

\begin{table}[H]
\centering 
\caption{Fitted parameter values according to Fig.~\ref{fig:wigaeus}.}\label{table:fitwigaeus}
\begin{tabular}{|l|ccccc|}\hline
 {\large\strut}&{\large\strut}$\vmax$&{\large\strut}$\pr$&{\large\strut}$D^{\mathrm{rest}}$&{\large\strut}$\qbr^{\mathrm{rest}}$&{\large\strut}$(\cbr,\calv,\cliv,\ctis)$ \\ 
{\large\strut}&{\large\strut}(mg/min/kg$^{\textrm{0.75}}$)&{\large\strut}(mg/min)&{\large\strut}(l/min)&{\large\strut}(\%)&{\large\strut}(mg/l)  \\ 
 \hline \hline 
{\large\strut}Anticipated &$0.31$ &$\leq 2$& $0$&$1$&$(0.0011,\ast,\ast,\ast)$ \\
{\large\strut}Optimized &$0.62$ &$0.19$& $0$&$0.43$&$(0.0016,0.0029,0.58,0.72)$ \\
\hline
\end{tabular}
\end{table}

Anticipated (literature) values of $\cb_0$ and $\tb$ compare favorably to those obtained by the aforementioned minimization, thus consolidating the physiological plausibility of the extracted estimates. In particular, the fitted value of the gas exchange location parameter $D^{\mathrm{rest}}=0$ agrees well with what is expected from the theoretical discussion in Section~\ref{sect:steadystates}.
The rate at which endogenous acetone forms in healthy adults as a result of normal fat catabolism is not known. However, from the study on acetone metabolism in lean and obese humans during starvation ketosis published by Reichard et al.~\cite{reichard1979} 
a rather conservative upper bound for $\pr$ can be derived to lie in the range of $2$~mg/min.
A diminished fractional bronchial perfusion $\qbr^{\mathrm{rest}}$ might reflect the fact that the assumption of diffusion equilibrium between the mucosa lining and the deeper vascularized sections of the airway wall is rather stringent, so that actually less acetone is transported away from the peripheral bronchial tract via the bloodstream. On the other hand $\qbr^{\mathrm{rest}}=0.01$ refers to the entire bronchial circulation rather than only the part contributing to anatomic right-to-left shunt as discussed in Section~\ref{sect:broex}. The fitted steady state value for the measured breath concentration $\cmeas(0)=\cbr(0)=0.0016$ mg/l during normal breathing at rest is only slightly higher than the observed levels spreading around $500$~ppb (corresponding to about 0.0011~mg/l)~\cite{schwarzace,setup}.\\

\begin{remark}\label{rem:farhiwig}
Fig.~\ref{fig:wigaeus} clearly illustrates the necessity of taking into account the conducting airways as an additional compartment for exchange of highly soluble substances. In particular, the observed data profiles are in sharp contrast to the classical Farhi inert tube description, predicting arterial blood concentrations to be directly proportional to (alveolar) breath concentrations. Contrarily, by taking into account pre-alveolar uptake as discussed in Section~\ref{sect:broex}, the accumulation of exogenous acetone in the systemic circulation is expected to be delayed due to the small contribution of bronchial blood flow to overall perfusion.
\end{remark}  

\begin{remark}
The population spread of the fitted parameters within the study cohort could be assessed, e.g., by a Bayesian~\cite{moerk2009} or mixed effects approach~\cite{kuhn2005}, which, however would be beyond the scope of this paper. Here, the major aim rather is to demonstrate the flexibility of the model in covering a wide spectrum of different experimental scenarios.
\end{remark} 

From an operational perspective, one may consider the results stated in the previous paragraphs as a model tuning procedure for resting conditions. Accordingly, in the remaining part of this paper both the stratified conductance parameter and the fractional bronchial blood flow during rest will be frozen at their fitted values $D^{\mathrm{rest}}=0$ and $\qbr^{\mathrm{rest}}=0.0043$, respectively. Moreover, due to the high value of the apparent Michaelis constant $\km$ as compared to the estimated acetone concentration $\lliv\cliv$ in hepatic venous blood, acetone clearance from the liver can safely be postulated to obey linear kinetics in the sequel, cf.~Equation~\eqref{eq:ratemet}. In particular, the extracted value of $\klin=\vmax/\km =0.0074$~l/kg$^{\textrm{0.75}}$/min is interpreted as an intrinsic property of acetone metabolism (with inter-individual variation being introduced by multiplication with the 0.75 power of body weight). Hence, $\met$ as in Equation~\eqref{eq:ratemet} will be treated as a known constant parameter in the sequel.

\subsection{Ergometer data sets and comparative evaluations}
\label{sect:ergo}

As has been indicated in the introduction, a primary motivation for this work was to develop a model elucidating the features of breath acetone behavior observed during moderate workload ergometer challenges. A representative profile corresponding to one \emph{single} normal healthy volunteer is shown in Fig.~\ref{fig:compfig}, cf.~\cite{setup}. As has been demonstrated there, end-tidal acetone levels in response to exercise generally resemble the profile of alveolar ventilation and inhalation volume, showing abrupt increases and drops in the range of 10~--~40\% at the onsets and stops of the individual workload periods, respectively (see also Fig.~\ref{fig:compfig}~(a), first panel). 
Similarly, a series of auxiliary experiments carried out by means of the same instrumental setup revealed that increasing tidal volume during rest results in increased breath acetone concentrations, while increasing respiratory frequency has a less pronounced impact, cf.~Fig.~\ref{fig:compfig}~(b). Both effects appear to support the hypothesis of Section~\ref{sect:broalvinter} that acetone exchange is strongly influenced by volume and speed of inhalation and hence suggest that any model not incorporating this mechanism will fail to reproduce the above results.
In particular, note that from the viewpoint of the classical Farhi description~\eqref{eq:farhi} the profile in Fig.~\ref{fig:compfig} is rather counter-intuitive. For instance, during hyperventilation acetone supply from the bloodstream will stay roughly constant, while a drastic increase in ventilation should enhance the dilution of alveolar air and would therefore be expected to (slightly) \emph{decrease} the corresponding breath concentration. 

From the above, we view exercise and hyperventilation scenarios as an interesting control setting for testing the physiological adequacy of the newly developed model. Moreover, these two scenarios provide a novel framework for contrasting the predictive power of distinct mechanistic frameworks put forward in the literature. Specifically, in the following we shall use the individual acetone behavior as depicted in Fig.~\ref{fig:compfig} as the basis for comparing the performance of Equations~\eqref{eq:bro}--\eqref{eq:tis} with the standard formulation due to Farhi
as well as with the physiological compartment model introduced by M\"{o}rk et al.~\cite{moerk2006}, which can be viewed as the current state of the art in acetone pharmacokinetic modeling. In particular, the latter model represents a first improvement over the Farhi description in that it has been found capable of adequately reproducing the bioaccumulation behavior presented in the previous subsection. M\"{o}rk et al.~also discuss several shortcomings of earlier models proposed in this context.

While the Farhi description is just a special case our model (by setting $\qbr=0$ and $D \to \infty$, cf.~Fig.~\ref{fig:model_hier}), the model of M\"{o}rk et al. will be re-implemented in slightly modified form by replacing the mass balance Equations~\eqref{eq:bro}--\eqref{eq:alv} of the respiratory tract with Equations (4)--(7) in~\cite{moerk2006}. The body compartments remain unchanged. Additional parameter values are taken from~\cite{moerk2006}. 

Letting $\cinh \equiv 0$ and $\cbr(0)=1.3$ $\mu$g/l, the initial compartment concentrations for all three models as well as the endogenous production rate $\pr$ can uniquely be determined by solving the associated steady state equations at $t_0=0$. This completely specifies the models by Farhi and M\"{o}rk et al..   
Contrarily, the response for the present model is computed by solving the ordinary least squares problem
\begin{equation}\label{eq:OLSexerc}\underset{{\ab}}{\mathrm{arg\,min}} \sum_{i=0}^N \big(y_i-\cbr(t_i)\big)^2,\quad \mathrm{s.t.}\;
\ab \geq 0 \quad \textrm{(positivity)}
\end{equation}
within the time interval $[t_0,t_N]=[0,t_{\mathrm{max}}]$. Here, $y_i=C_{\mathrm{measured},i}$ is the observed acetone concentration at time instant $t_i$ and $\alpha=(\alpha_1,\ldots,\alpha_{m-1})$ is the coefficient vector for the piecewise linear function
\begin{equation}\label{eq:Dhat}
\hat{D}(t):=\sum_{j=1}^{m-1} \alpha_j S_j(t),\quad S_j(t):=\left\{\begin{array}{ll}
\frac{t-s_{j-1}}{s_j-s_{j-1}}& t \in [s_{j-1},s_j]\\
\frac{s_{j+1}-t}{s_{j+1}-s_j}& t \in [s_j,s_{j+1}]\\
0 & \textrm{otherwise,} \end{array} \right.
\end{equation} 
used for deriving an approximation of the time-varying stratified conductance parameter $D \in [0,\infty)$ on $m$ subintervals covering the time span $[t_0=s_0,t_{\mathrm{max}}=s_m]$. Here the nodes $s_j$ are chosen to result in an equidistant partition of about $0.5$~min, cf.~Fig.~\ref{fig:compfig}, third panel.
For simulation purposes the measured physiological functions are converted to input function handles $\ub$ by applying a local smoothing procedure to the associated data and interpolating the resulting profiles with splines. The aforementioned optimization problem was solved as described in Section~\ref{sect:wigaeus}. Fig.~\ref{fig:compfig} summarizes the results of these calculations.

The visually good fit can formally be assessed by residual analysis. Plots of the resulting residuals versus time and versus model predictions clearly exhibit random patterns, suggesting that the assumptions of i.i.d., homoscedastic additive measurement errors underlying ordinary least squares methodology are reasonable~\cite{banksbook}. Furthermore, no statistically significant autocorrelation among the residuals or cross-correlation between the residuals and the measured inputs could be detected, indicating that the model has picked up the decisive dynamics underlying the data.

\begin{figure}[H]
\vspace{-0.5cm}
\centering
\begin{tabular}{c}
\hspace*{-1.2cm}
\includegraphics[height=14cm]{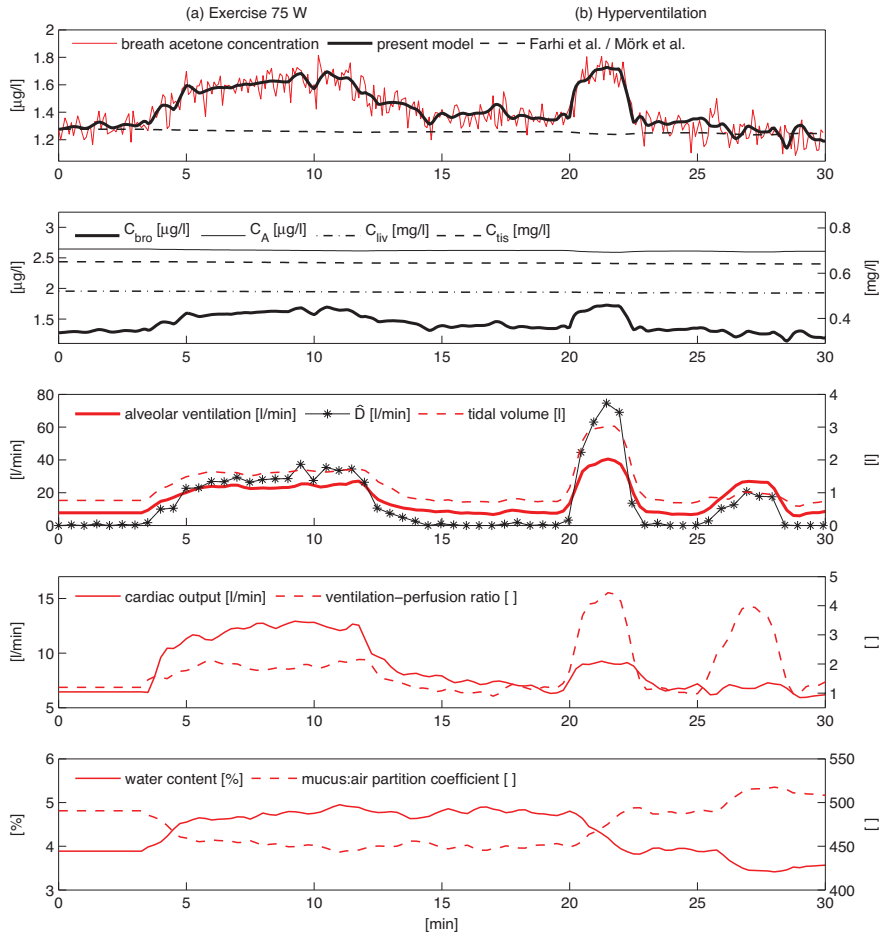}
\end{tabular}
\caption{Typical profile of end-exhaled acetone concentrations in response to the following physiological regime: rest (0-3~min), ergometer challenge at 75~W (3-12~min), rest (12-20~min), hyperventilation with increased tidal volume (20-22~min), rest (22-25~min), high-frequency hyperventilation (25-28~min), rest (28-30~min). Data correspond to one representative healthy male volunteer from the study cohort in~\cite{setup}. Measured or derived quantities according to the experimental setup are indicated in red, while black tracings correspond to simulated variables. \emph{Panel~1}: Measured acetone concentrations in end-tidal breath and associated model predictions; \emph{Panel~2}: Simulated compartment concentrations according to Equations~\eqref{eq:bro}--\eqref{eq:tis}; \emph{Panel~3--5}: Measured physiological parameters according to Table~\ref{table:measparams}. The mucus:air partition coefficient $\lmuca$ and the time-varying stratified conductance parameter $\hat{D}$ are derived from Equation~\eqref{eq:meanhen} and~\eqref{eq:Dhat}, respectively.}\label{fig:compfig}
\end{figure}

While all three models will describe the steady state at rest, only Equations~\eqref{eq:bro}--\eqref{eq:tis} tolerably capture the entire observable dynamical behavior. Particularly, the predictions resulting from the models due to Farhi and M\"{o}rk et al.~are almost indistinguishable within this experimental regime and both will depart only slightly from the initial equilibrium states. In contrast, the newly proposed model indeed appears to incorporate the decisive physiological mechanisms underlying the step-shaped dynamics of breath acetone concentrations in response to constant load exercise and hyperventilation. Correspondingly, according to the second panel in Fig.~\ref{fig:compfig}, measured (i.e., bronchial) levels during rest and tidal breathing markedly differ from the alveolar ones due to the effective diffusion barrier between the two spaces represented by a value of $D^{\mathrm{rest}} \approx 0$, cf.~Section~\ref{sect:wigaeus}. As soon as tidal volume and/or respiratory frequency are increased, this barrier will vanish to some extent according to the processes discussed in Section~\ref{sect:steadystates}, thus causing $\cbr=\cmeas$ to approach a value closer to the alveolar concentration $\calv$. 
In contrast, $\calv$ itself as well as the breath acetone profiles simulated by means of the other two models remain relatively constant. This is consistent with the behavior expected from the Farhi formulation, predicting a minimal sensitivity of the alveolar concentrations with respect to fluctuations of the ventilation-perfusion ratio for highly soluble trace gases. In other words, the major part of short-term variability observable in breath acetone concentrations during free breathing can be attributed to airway gas exchange, with minute changes in the underlying blood and tissue concentrations. In particular, note that with the present model the observed acetone dynamics can be captured by assuming a constant endogenous production rate $\pr \approx 0.17$~mg/min. The above-mentioned reasoning appears to agree with previous observations in the literature, where the excretion of acetone has been demonstrated to increase during moderate exercise~\cite{schrikker1989}.

\begin{remark}The high degree of interplay between $\hat{D}$ and $(\qalv,\vt)$ discernible in the third panel of Fig.~\ref{fig:compfig} suggests that the time-dependency of $D$ can essentially be captured using these two respiratory variables. Different heuristic relationships might be investigated in this context, see Equation~\eqref{eq:D} for instance. Correspondingly, by repeating the optimization procedure in~\eqref{eq:OLSexerc} with $\ab$ being replaced by $(\dke,\dkz)$ we find that the associated model response again is in good agreement with the observed data. The practical identifiability of these two parameters at their optimized values $\dke=14.9$ and $\dkz=0.76$ was confirmed along the lines of Remark~\ref{rem:ident}. Hence, the parameterization in Equation~\eqref{eq:D} might be used to reduce the originally infinite dimensional estimation problem for $D$ to two degrees of freedom.
\end{remark}

While in Section~\ref{sect:identif} it has been confirmed that the model is structurally locally observable, we stress the fact that data corresponding to exercise or hyperventilation tests as presented above will usually not allow for a joint \emph{numerical} estimation of some of the above variables. For perspective, it should be clear from the approximately constant profile of $\cliv$ in Fig.~\ref{fig:compfig} that $\pr$ and $\met$ can hardly be assessed simultaneously, as both values are coupled via a single steady state equation associated with Equation~\eqref{eq:liv}. Using similar reasoning, smaller values of $\qbr^{\mathrm{rest}}$ (leading to larger steady state concentrations $\calv(0)$) can be compensated for by smaller constants $k_{\mathrm{diff},j}$, giving rise to an almost identical model output. Again, such a situation might formally be investigated by means of practical identifiability techniques (cf.~Appendix~\ref{sect:krener}), e.g., by calculating pairwise correlation coefficients between the sensitivities of the model output with respect to the parameters under scrutiny. In the present context, values of these indices associated with $(\met,\pr)$ as well as $(\dke,\qbr^{\mathrm{rest}})$ and $(\dkz,\qbr^{\mathrm{rest}})$ are close to $1$ or $-1$, thereby revealing a substantial degree of collinearity between the corresponding sensitivities and providing strong indications that a proper estimation of any of these pairs on the basis of moderate exercise challenges as above cannot be anticipated. Such identifiability issues can only be circumvented by designing (multi-)experimental regimes guaranteeing a sufficiently large and independent influence of all parameters to be estimated. 

\section{Discussion and critical remarks}

The main intention of this section is to critically review and clarify some of the assumptions underlying the model derivation as well as to indicate some potential improvements of the present formulation.

\paragraph*{Diffusion equilibrium in the bronchial tract.\;} While it is commonly agreed upon that the transport of inert gases between the alveolar space and end-capillary blood is perfusion-limited~\cite{wagner2008}, a similar premise in the case of airway gas exchange remains less certain. Experimental and theoretical evidence appears to favor the view that although a diffusion equilibrium might be attained at the air-mucus interface~\cite{tsu1988,kumagai1999}, the bronchial epithelial tissue can constitute an effective diffusion barrier between the mucus lining and the bronchial circulation. The magnitude of this diffusional resistance is probably substance-specific, with an inverse relation to molecular weight~\cite{swenson1992}. In this sense, the fractional bronchial blood flow $\qbr$ has to be interpreted as the effective perfusion of those mucosal tissue layers for which an instantaneous equilibrium with air can be achieved~\cite{kumagai2000}. Specifically, $\qbr$ might differ for distinct compounds.
While animal models seem to support the assumption of a complete equilibration between air stream and bronchial circulation in the special case of acetone~\cite{swenson1992}, it might be necessary to include a diffusion limitation between these two compartments in order to extend the validity of the model over a wider range of highly water soluble VOCs. The statistical significance of such generalizations can then be assessed by employing residual based comparison techniques for nested models as described in~\cite{banksbook,banks1990}. However, at the current stage of research and given the limited data on the behavior of breath trace gases having similar physico-chemical characteristics like acetone, we prefer to maintain a parameterization as parsimonious as possible. 

\paragraph*{Continuous ventilation and temperature dependence.\;}

Due to the tidal nature of breathing, bronchial as well as alveolar gas concentrations will vary throughout the breathing cycle, following a roughly periodical pattern during normal breathing. These variations can be captured by considering two separate mass balance systems, describing the dynamics of the associated concentrations for each inhalation and exhalation phase, respectively~\cite{martonen1982,kumagai2000,anderson2003}. While such microscopic formulations are of paramount importance for resolving events within one individual respiratory cycle, when looking at the mid- to long-term behavior of breath VOCs we are rather interested in the global dynamics of the averaged compartmental concentrations. This approach leads to models of continuous ventilation with a unidirectional gas stream and has the enormous operational advantage of reducing the model structure to one single mass balance system. 

Variable temperature distributions in the bronchial compartment are represented by a single mean temperature $\bar{T}$ according to Equation~\eqref{eq:temp}, thereby lumping together our ignorance regarding the exact temperature profile along the airways. Different functional relationships might be investigated here.\\ 

In summary, this paper introduces a novel compartmental description of pulmonary gas exchange and systemic distribution of blood-borne, highly blood and water soluble VOCs, which faithfully captures experimentally determined end-tidal acetone profiles for normal healthy subjects during free breathing in distinct physiological states. Particularly, the model has been tested in the framework of external exposure as well as exercise scenarios and illuminates the discrepancies between observed and theoretically predicted blood-breath ratios of acetone during resting conditions, i.e., in steady state. In this sense, the present formulation provides a good compartmental perspective of acetone exhalation kinetics and is expected to contribute to a better understanding of distribution, transport, biotransformation and excretion processes of acetone in different functional units of the organism as well as their impact on the observed breath concentration.

While we are well aware of the fact that the number of data sets used for model validation is relatively small, exposure and exercise regimes currently are the only published experimental settings covering non-steady state acetone behavior in humans. Further validation will thus have to await additional experimental efforts. In this context, preliminary tests conducted with the intention of extending the range of applicability for the presented formulation to the framework of isothermal rebreathing show promising results~\cite{King2010b}.
 
The emphasis of this work has been laid on deriving a sound mathematical formulation flexible enough to cover a wide spectrum of possible VOC behavior, while simultaneously maintaining physiological plausibility as well as a clear-cut interpretation of the involved parameters. 
Care has been taken to keep the parameterization as parsimonious as possible, thereby constructing a novel ``minimal'' model respecting fundamental physiological boundary conditions, such as boundedness of the associated trajectories and the existence of a globally asymptotically stable equilibrium state. While a complete sensitivity and practical identifiability analysis was beyond the scope of this paper and has to be matched to the particular experimental framework in which the model will be used, general concepts from structural and practical identifiability have been exploited in order to provide some indications regarding the information content of the observable breath level with respect to the endogenous situation.

\begin{acknowledgements}
We are indebted to the reviewers for several helpful suggestions.
Julian King is a recipient of a DOC fellowship at the
Breath Research Institute of the Austrian Academy of Sciences.
The research leading to these results has received funding
from the European Communitys Seventh Framework Programme
(FP7/2007-13) under grant agreement No.~217967.
We appreciate funding from the Austrian Federal Ministry
for Transport, Innovation and Technology (BMVIT/BMWA,
project~818803, KIRAS). Gerald Teschl and Julian King acknowledge support from the Austrian Science Fund (FWF) under
Grant No.~Y330. We greatly appreciate the generous
support of the government of Vorarlberg and its governor Landeshauptmann
Dr. Herbert Sausgruber.
\end{acknowledgements}


\appendix
\section{Physical preliminaries}
The following paragraphs briefly summarize the fundamental physical principles underlying the present model derivation. Further discussion largely follows standard textbooks on thermodynamics and biological mass transport, see~\cite{truskeybook} for instance.
\subsection{Diffusion and Henry's Law}
Fick's first law of diffusion states that the net flux $J$ quantifying diffusional transport of a gas species $M$ between two spatially homogeneous control volumes (i.e., compartments in which $M$ can be postulated to behave uniformly) $A$ and $A'$ is proportional to the associated difference in partial pressure $P_M$, i.e.,
\begin{equation}\label{eq:fick}J=k_{\mathrm{diff}}(P_{M,A}-P_{M,A'}),\end{equation}
where the non-negative diffusion constant $k_{\mathrm{diff}}$ has dimensions of amount (in mol) divided by pressure and time. 
If the gas under scrutiny is in dissolved state rather than in gas phase, $P_M$ is referred to as gas tension,
defined as the partial pressure that would be exerted by the gas above the liquid if both were in equilibrium (e.g., oxygen tension in blood). In the following we will omit the subscript $M$. Equation~\eqref{eq:fick} can also be expressed via a difference in concentrations.
We will distinguish the following three cases:

(i). Firstly, if $A$ as well as $A'$ represent a gaseous medium and $M$ can be treated as an ideal gas (which will be a general premise in the following) then -- by the ideal gas law -- its partial pressure $P$ can be written as
\begin{equation}\label{eq:gaslaw}P_A=C_A R T,\end{equation}
and analogously for $P_{A'}$. Here, $C_A$ is the molecular concentration in $A$ while $R$ and $T$ denote the gas constant and absolute temperature, respectively.
Hence,
\begin{equation}J=k_{\mathrm{diff}}(C_{A}RT-C_{A'}RT)=\tilde{k}_{\mathrm{diff}}(C_{A}-C_{A'}),\end{equation}
where $\tilde{k}_{\mathrm{diff}}$ is defined by the above equation.

(ii). Similarly, if both $A$ and $A'$ are liquids, we can make use of Henry's law stating that the amount of a gas $M$ that can dissolve in a liquid under equilibrium conditions is directly proportional to the respective tension, viz.,
\begin{equation}\label{eq:henry}C=HP.\end{equation}
Here the solute- and solvent-specific Henry constant $H=H(T)$ is inversely related to temperature.
Substituting the last expression into Equation~\eqref{eq:fick} yields
\begin{equation}\label{eq:liquid}J=k_{\mathrm{diff}}\lk\frac{C_{A}}{H_A}-\frac{C_{A'}}{H_{A'}}\rk.\end{equation}

(iii). Finally, let us assume that $A'$ represents a liquid, while $A$ is a gas volume. The prototypic example for this situation in the present context is the blood-gas interface separating the respiratory microvasculature and the alveoli. In this case by combining~\eqref{eq:gaslaw} and~\eqref{eq:henry} we arrive at
\begin{equation}\label{eq:gas}J=k_{\mathrm{diff}}RT\lk C_{A}-\frac{C_{A'}}{H_{A'}RT}\rk=\tilde{k}_{\mathrm{diff}}\lk C_{A}-\frac{C_{A'}}{H_{A'}RT}\rk.\end{equation}
In case of a diffusion equilibrium, i.e., $J=0$, we deduce from Equation~\eqref{eq:liquid} that 
\begin{equation}\label{eq:partcoef}\frac{C_{A'}}{C_{A}}=\frac{H_{A'}}{H_{A}}=:\lambda_{A':A},\end{equation}
where the positive dimensionless quantity $\lambda_{A':A}$ is the so-called partition coefficient between $A'$ and $A$ (e.g., $\lambda_{\mathrm{blood:fat}}$). Analogously, its reciprocal value is denoted by $\lambda_{A:A'}$.
Similarly, from~\eqref{eq:gas} the ratio between the liquid phase concentration and gas phase concentration of $M$ in equilibrium is given by
\begin{equation}\frac{C_{A'}}{C_{A}}=H_{A'}RT.\end{equation}
The expression $H_{A'}RT$ is again denoted by $\lambda_{A':A}$ (e.g., $\lambda_{\mathrm{blood:air}}$) and represents a measure of the solubility of the gas
$M$ in the solvent $A'$. Hence, $\lambda_{A':A}$ allows for a classification of $M$ as low or highly soluble with respect to $A'$.

\subsection{Compartmental mass transport}\label{sect:compmod}
As above, consider two homogeneous compartments $A$ and $A'$ with volumes $V_A$ and $V_{A'}$, respectively.
Gas dynamics within these two compartments are governed by conservation of mass, stating that the rate at which the 
compartmental amount (viz., $V_A C_A$) of $M$ changes per time unit equals the rate of mass transfer into the compartment less the rate of removal from the compartment. For instance, the process of diffusion between $A$ and $A'$ as discussed before can be seen as a special case of this scheme with $J$ (expressed in terms of concentrations) combining both diffusional inflow and outflow. Hence, if $A$ and $A'$ interact by diffusion and if both are affected by additional input rates $\dot{I}$ and effluent output rates $\dot{O}$ (which might also comprise possible sources or sinks, respectively), then, assuming that the respective compartment volumes remain constant over time, the associated mass balances read
\begin{equation}\label{eq:diff1}V_A\frac{\di C_A}{\di t}=-J+\dot{I}_A-\dot{O}_A\end{equation}
and
\begin{equation}\label{eq:diff2} V_{A'}\frac{\di C_{A'}}{\di t}=J+\dot{I}_{A'}-\dot{O}_{A'},\end{equation}
respectively. 

If diffusion is fast compared to the other dynamics influencing the system, $C_A$ as well as $C_{A'}$ will instantaneously approach their steady state values and a permanent diffusion equilibrium can be assumed to hold between $A$ and $A'$. 
In this case we may replace $C_{A'}$ in~\eqref{eq:diff2} with $C_A\lambda_{A':A}$ (cf.~Equation~\eqref{eq:partcoef}) and deduce that the ODE system above simplifies to the one-dimensional differential equation
\begin{equation}\label{eq:diffcomb} \big(V_{A'}\lambda_{A':A}+V_{A}\big)\frac{\di C_{A}}{\di t}=\dot{I}_{A'}-\dot{O}_{A'}+\dot{I}_A-\dot{O}_A.\end{equation}
The factor $\tilde{V}_A:=(V_{A'}\lambda_{A':A}+V_{A})$ can be interpreted as ``effective'' volume of the combined compartment.
 
Equation~\eqref{eq:diffcomb} characterizes a \emph{perfusion-limited} mass transfer between $A$ and $A'$. For perspective, there is strong evidence that in normal healthy persons the alveolar exchange if inert gases is a perfusion-limited process. The branching capillary network of the respiratory microvasculature will generally promote a fast equilibration between end-capillary blood $A'$ and free gas phase $A$. This is likely to hold true even under moderate exercise conditions~\cite{wagner2008}. From a modeling perspective, a commonly encountered error in this context is to neglect the contribution of $V_{A'}\lambda_{A':A}$ to the effective alveolar volume, as end-capillary blood volume $V_{A'}\approx 0.15$~l is argued to be much smaller than the normal lung volume $V_{A}\approx 3$~l. Note, however, that this is only acceptable for low soluble inert gases with blood:air partition coefficient less than 1.
Similarly, for most tissue groups, owing to their high degree of vascularization the intracellular space $A$ can be assumed to rapidly equilibrate with the extracellular space $A'$ (including the vascular blood and the interstitial space). In this case, a venous equilibrium is said to hold, i.e., the concentration $C_{A'}$ of $M$ in blood leaving the tissue group is related to the actual tissue concentration $C_A$ via Equation~\eqref{eq:partcoef}.\\

Note that the assumption of well-mixing is central to compartmentalization, and can often only be justified heuristically, either by examining factors that contribute to rapid distribution (e.g., convection) or by considering small volumes. If heterogeneity within a compartment is expected to be substantial or the main focus is on the spatial distribution of $M$, general mass transport equations leading to PDEs may be employed~\cite{truskeybook}.

\section{Some fundamental model properties}\label{sect:apriori}
Equations~\eqref{eq:bro}--\eqref{eq:tis} can be written as a time-varying linear inhomogeneous ODE system
\begin{equation}\label{eq:sysc}
\dot{\cb}=A(\ub,\tb)\cb+\bb(\ub,\tb)=:\gb(\ub,\cb,\tb)
\end{equation}
in the state variable $\cb:=(\cbr,\calv,\cliv,\ctis)^T$, which is dependent on a time-inde\-pendent parameter vector $\tb$ as well as on a vector
$\ub:=(\qalv,\qc,\vt,\lmuca(\cw),\cinh)$ of bounded, non-negative functions lumping together all measured variables. The associated measurement equation reads
\begin{equation}y=\cmeas=(1,0,0,0)\,\cb=:h(\cb).\end{equation}
In the present section we are going to discuss some qualitative properties of the system presented.
Firstly, as the system is linear, for any given initial condition $\cb(0)$ there is a unique global solution.
Moreover, $A$ is a Metzler matrix and $\bb \geq 0$. Consequently $c_i=0$ implies that $\dot{c}_i \geq 0$ for every component and it follows that the trajectories remain non-negative.
Thus,~\eqref{eq:sysc} constitutes a positive system with the state $\cb$ evolving within the positive orthant $\R^n_{>0}:=\{\cb\,|\,\cb \in \R^n,\;c_i>0,\;i=1,\ldots,n\}$, where $n=4$. 

\begin{proposition}
All solutions of~\eqref{eq:sysc} starting in $\R^n_{>0}$ remain bounded.
\end{proposition}

\begin{proof}
This can be shown by considering the total mass $m:=\sum \tilde{V}_i c_i \geq 0$ and noting that 
\begin{equation}\label{eq:massderiv}\dot{m}=\pr-\met\lliv\cliv+\qalv(\cinh-\cbr).\end{equation}
Taking into account positivity of the solutions and the involved parameters this shows that $\cliv$ is bounded from above for bounded $\cinh$, since the assumption that $\cliv$ is unbounded yields a contradiction. Analogously, $\cbr$ and $\calv$ can be shown to be bounded by considering $\tilde{m}:=\vbro\cbr+\valv\calv+\vtis\ctis$ and similarly for $\ctis$. 
\end{proof}

Furthermore, it can be proven that under physiological steady state conditions, i.e., for constant $\ub$, the above system has a globally asymptotically stable equilibrium $\cb^e:=-A^{-1} \bb$. To this end it suffices to show that the time-invariant matrix $A$ is Hurwitz, i.e., all real parts of the associated eigenvalues are negative. 

\begin{proposition}
Suppose $\ub$ is time-independent. Then the real parts of all eigenvalues of $A(\ub,\tb)$ are non-positive. They
are strictly negative if $\det(A(\ub,\tb))\ne 0$.

Moreover, $\det(A(\ub,\tb)) = 0$ if and only if either $\qc$ vanishes or two of the quantities $\qalv, \qliv, \met$ vanish.
\end{proposition}

\begin{proof}
For this it is sufficient to confirm that $A$ is diagonally dominant, i.e., there exists a vector $\zb > 0$ such that the row vector $\zb^T A$ is non-positive. This is a simple consequence of the fact that in such a case $A$ can be shown to be similar to a matrix $\tilde{A}$ which has the claimed property. Indeed, if we define the diagonal matrix $U:=\mathrm{diag}(z_1,\ldots,z_n)$ and set $\tilde{A}:=UAU^{-1}$, it holds that 
\begin{equation}\tilde{a}_{jj}+\sum_{i\neq j}\abs{\tilde{a}_{ij}}=a_{jj}+\sum_{i\neq j}\frac{a_{ij}z_i}{z_j}=\frac{1}{z_j}(\zb^T A)_j \leq 0\end{equation}
for all $j$ since $A$ is Metzler. Hence, the first claim will follow from Gershgorin's circle theorem. The
required vector $\zb$ follows from
\[
(\vbro,\valv,\vliv,\vtis)A=(-\qalv,0,-\met\lliv,0).
\]
Moreover, the only case when the real part of an eigenvalue can vanish is when it lies on the boundary of a circle touching the imaginary axis,
that is, when an eigenvalue is zero. This proves the first part.

To show the second part one verifies that
\begin{equation}
\det(A) = \qc \big((\vartheta_1 \qalv + \vartheta_2 \met) \qliv \qc + \vartheta_3 \qalv \met\big)
\end{equation}
with $\vartheta_j >0$.
\end{proof}

This completes the minimal set of properties which necessarily should to be satisfied in any valid model of concentration dynamics. In particular, global asymptotic stability for the autonomous system ensures that -- starting from arbitrary initial values -- the compartmental concentrations will approach a unique equilibrium state once the physiological inputs $\ub$ affecting the system are fixed. Obviously, this is of paramount importance when aiming at the description of processes exhibiting pronounced steady states (which in the context of breath gas analysis corresponds, e.g., to the situation encountered during rest or constant workload~\cite{setup}) and is a prerequisite for orchestrating reproducible experiments. Hence, in the context of VOC modeling, any approach not incorporating this property will lack a fundamental characteristic of the observed data.

\begin{remark}
For perspective, we stress the fact that if the rate of metabolism in Equations~\eqref{eq:bro}--\eqref{eq:tis} is described by saturation rather than linear kinetics, i.e., if the term $\met\cliv\lliv$ in Equation~\eqref{eq:liv} is replaced by Equation~\eqref{eq:metmenten}, the system becomes essentially nonlinear. However, all conclusions drawn so far remain generally valid. Firstly, by applying the Mean Value Theorem note that the Michaelis-Menten term above is Lipschitz and so again a global Lipschitz property holds for the right-hand side of the resulting ODE system. Positivity and boundedness of the solutions for arbitrary inputs $\ub$ can be established in analogy with the arguments in the last paragraphs.
Effectively, in order to show boundedness from above, we first note that Equation~\eqref{eq:massderiv} now reads 
\begin{equation}\dot{m}=\pr-\frac{\vmax\cliv\lliv}{\km+\cliv\lliv}+\qalv(\cinh-\cbr)
\end{equation}
and as a result $\cbr$ is unconditionally bounded if $\qalv > 0$. From Equation~\eqref{eq:bro} it then follows that $\calv$ must be bounded and subsequently $\cven$ (and hence $\cliv$ as well as $\ctis$) are bounded using Equation~\eqref{eq:alv}. In the case $\qalv=0$, the inequality 
\begin{equation}\label{eq:bounded}\pr\leq\vmax \end{equation} is necessary and sufficient for boundedness. Necessity is easily deduced from the fact that $\pr>\vmax$ implies $\dot{m}>0$ and consequently $m$ (i.e., at least one component $c_i$) is unbounded. Conversely, Equation~\eqref{eq:bounded} ensures that $\cliv$ is bounded and hence the same arguments as in the linear case apply.\\ 
In order to establish the existence of a globally asymptotically stable equilibrium point for fixed $\ub$ we adopt a result on monotone systems taken from Leenheer et. al.~\cite{leenheer2007}, see also~\cite{jifa1994}. For monotone systems in general with applications to biological systems we refer to~\cite{smithmonotone,angeli2003}. Firstly, note that the steady state relation associated with Equation~\eqref{eq:liv} now becomes
\begin{equation}0=\cliv^2+a_1\cliv +a_0\end{equation}
with $a_0<0$. Thus, there exists a unique positive steady state solution for the liver concentration $\cliv$ and consequently the same can be confirmed to hold true for the other components of $\cb$ as well. As a consequence, it follows that the model has a unique equilibrium point in the non-negative orthant $\mathcal{X}:=\R^n_{\geq 0}$.
Moreover, note that for the right-hand side $\gb$ of the underlying ODE system (cf.~Equation~\eqref{eq:sysc}) it holds that
\begin{equation}\frac{\partial g_i}{\partial c_j} \geq 0\;\textrm{for}\;i\neq j,\;i,j \in \{1,\ldots,n\},\end{equation}
i.e., the system is cooperative. This term stems from the fact that a given component is positively affected by the remaining ones. For linear systems, an equivalent characterization is that the corresponding system matrix $A$ is Metzler. A well-known result on systems of this type asserts that the associated semiflow $\Phi:\R_{\geq 0}\times \R^n \to \R^n,\;(t,\cb_0)\mapsto \Phi_t(\cb_0):=\cb(t)$ is monotone with respect to the natural (componentwise) partial order on $\R^n$ given by
\[\cb \leq \zb\;\textrm{if and only if}\; c_i \leq z_i,\; i \in \{1,\ldots,n\}.\]
That is, $\Phi$ preserves the order of the initial conditions (see~\cite{smithmonotone}, Prop.~3.1.1), i.e., for $\cb_0,\tilde{\cb}_0 \in \R^n$ the condition $\cb_0 \leq \tilde{\cb}_0$ implies that $\Phi_t(\cb_0) \leq \Phi_t(\tilde{\cb}_0)$ for $t \in [0,\infty)$.
As a third requirement, since all trajectories are bounded and the system evolves within the closed state space $\mathcal{X} \subset \R^n$ it follows that for every $\cb \in \mathcal{X}$ the corresponding semi-orbit $O(\cb):=\{\Phi_t(\cb),\;t\geq 0\}$ has compact closure. 
In summary, we are now in the situation to apply Theorem~5 of~\cite{leenheer2007}, which asserts that the aforementioned properties are sufficient for the unique equilibrium point in $\mathcal{X}$ to be globally attractive. 
\end{remark}

\section{A primer on local observability}\label{sect:krener}
For the convenience of the reader, the following section serves as a self-contained overview of the concept of (local) observability for analytic, single-output state-space models of the form
\begin{align}\label{eq:auton}\tag{$\Sigma$}
\begin{split}
\dot{\xb} &= \gb(\xb,\ub),\quad \xb(0)=\xb_0,\\
y &= h(\xb,\ub).
\end{split}
\end{align}
Here it is assumed that $\xb(t)$ evolves within an open set $\mathcal{X} \subseteq \R^n$, $h:\mathcal{X} \to \R$ is real analytic and $\gb: \mathcal{X} \to \R^n$ is an analytic vector field. Furthermore, it is required that a unique solution exists, e.g., by assuming that $\gb$ is globally Lipschitz on $\mathcal{X}$. Since $\gb$ is assumed to be analytic, the solution will be analytic with respect to both $t$ and $\xb_0$.
The output $y=y(t,\xb_0)$ is viewed as a function of the initial condition $\xb_0$. It will be analytic as well. 

We specialize to the case where the inputs are constant and hence can be interpreted as additional known parameters of the system. In this case~\eqref{eq:auton} can be seen as a valid description for many biological processes under constant measurement conditions, thus representing a sufficient framework for the type of models and experiments considered here. Further exposition follows the treatment in~\cite{krener,nijmeier,sontag} and~\cite{anguelova}. \\

Loosely spoken, local observability characterizes the fact that $\xb_0$ can be instantaneously distinguished from any of its neighbors by using the system response. This is made explicit in the following definition.
\begin{definition}
Let $V \subseteq \mathcal{X}$. Two initial conditions $\xb_0,\tilde{\xb}_0 \in V$ are $V$-indistinguishable ($\xb_0 \sim_V \tilde{\xb}_0$) if for every $T > 0$ such that the corresponding trajectories remain in $V$ it holds that \[y(t,\xb_0) = y(t,\tilde{\xb}_0)\]
for all $t \in [0,T]$. A system~\eqref{eq:auton} is called locally observable at $\xb_0 \in \mathcal{X}$, if there exists a neighborhood $V \subseteq \mathcal{X}$ of $\xb_0$ such that for all points $\tilde{\xb}_0$ in $V$ the relation $\xb_0 \sim_V \tilde{\xb}_0$ implies that $\xb_0 = \tilde{\xb}_0$.
\end{definition}
In other words, local observability characterizes the fact that $\xb_0 \mapsto y(.,\xb_0)$ is injective for arbitrarily small times $T$. For the sake of illustration, consider the system
\begin{equation}
\begin{pmatrix}\dot{x}_1 \\ \dot{x}_2 \\\dot{x}_3 \end{pmatrix}=\begin{pmatrix}-(x_2+x_3)x_1\\0\\0\end{pmatrix},\quad y=x_1,
\end{equation}
then
\[y=x_{1,0}\exp(-(x_{2,0}+x_{3,0})t).\]
Evidently, this toy model is not locally observable as the slightly perturbed initial conditions $\tilde{x}_{2,0}:=x_{2,0}+\varepsilon$ and $\tilde{x}_{3,0}:=x_{3,0}-\varepsilon$ result in exactly the same response $y$ for any $\varepsilon > 0$. If we interpret $x_1$ and $x_2$ as constant parameters, the above situation corresponds to a typical case of over-parameterization, i.e., adding $x_3$ does not provide a more detailed description of the observable output. On the other hand, from the reversed viewpoint of estimation, this means that we will never be able to simultaneously estimate both $x_{2,0}$ and $x_{3,0}$ from the available process data but can only assess their sum. 

\begin{remark}An equivalent characterization of this fact can be given in the context of sensitivities, i.e., by considering the partial derivatives of the output with respect to the parameters. Indeed, we immediately see that
\[\frac{\partial y}{\partial x_{2,0}}  = \frac{\partial y}{\partial x_{3,0}},\] 
i.e., the sensitivities are linearly dependent, which is well-known to preclude a joint estimation of both parameters, see~\cite{cobelli1980,jac1990}. Intuitively, the reason for this is that from the the Taylor expansion of first order
\[y(\tilde{x}_{2,0},\tilde{x}_{3,0})\approx y(x_{2,0},x_{3,0})+\frac{\partial y}{\partial x_{2,0}}\varepsilon-\frac{\partial y}{\partial x_{3,0}}\varepsilon=y(x_{2,0},x_{3,0}),\]
which again confirms that any small change in $x_{2,0}$ can be compensated for by an appropriate change in $x_{3,0}$ to yield almost the same output. 
A rigorous statement of this fact in the general unconstrained case is given in~\cite{beckbook}. In contrast, as will be illustrated henceforth, local observability guarantees linear independence of the derivatives of $y$ with respect to $\xb_0$. This is also a necessary requirement for the broad class of numerical parameter estimation schemes based on the minimization of a given cost functional using first-order information.
\end{remark}
 
In more complex models, where an analytical expression for $y$ can usually not be derived, we have to resort to differential-geometric methods for investigating local observability. These tools will only be developed to the extent necessary for stating the final observability criterion, but are usually defined in a much more general framework than the one presented here.\\
For this purpose, first note that using the chain rule, the time derivative of the output equals the Lie derivative $L_{\gb} h$ of the output function $h$ with respect to the vector field $\gb$, i.e.,
\[\dot{y}= (\nabla h) \cdot \gb =: L_{\gb} h\]
which again is analytic. Here, $\nabla = (\frac{\partial}{\partial x_1},\dots,\frac{\partial}{\partial x_n})$ denotes the gradient and $\cdot$ denotes the scalar product in $\R^n$.
Similarly, by iteratively defining
\[L_{\gb}^{(k)}h=L_{\gb}^{\phantom{(}} L_{\gb}^{(k-1)}h, \quad L_{\gb}^{(0)}h=h,\]
we immediately see that
\[y^{(k)}(\xb_0):= \frac{\partial^k}{\partial t^k}  y(t,\xb_0) \Big|_{t=0} = L_{\gb}^{(k)}h(\xb_0).\]
By our analyticity assumption on the system~\eqref{eq:auton}, the system response $y$ is analytic can hence be
represented by a convergent power (Lie) series around $t=0$:
\begin{equation}\label{eq:lieseries}
y(t,\xb_0)=\sum_{k=0}^{\infty}\frac{t^k}{k!} y^{(k)}(\xb_0) =\sum_{k=0}^{\infty}\frac{t^k}{k!}L_{\gb}^{(k)}h(\xb_0).
\end{equation}
As a result, the relation $\xb_0 \sim_V \tilde{\xb}_0$ is characterized by the fact that for all $k \geq 0$ it holds that
\[L_{\gb}^{(k)}h(\xb_0) = L_{\gb}^{(k)}h(\tilde{\xb}_0).\]
This motivates the following two definitions. The linear space over $\R$ of all iterated Lie derivatives
\[\mathcal{O}(\xb_0) :=\mathrm{span}_{\R} \{ L_{\gb}^{(k)}h(\xb_0),\;k \geq 0 \},\]
is called the observation space of the model.
Moreover, $\nabla \mathcal{O}:=\{\nabla H,\;H\in \mathcal{O}\}$ denotes the associated observability co-distribution.
The famous Hermann-Krener criterion~\cite{krener} now states that the model is locally observable at $\xb_0$ if 
\begin{equation}\label{eq:obscrit}\mathrm{dim}\nabla \mathcal{O}(\xb_0)=n.\end{equation}

The sufficient condition~\eqref{eq:obscrit} is equivalent to finding $n$ indices $k_j$ such that the gradients $\nabla L_{\gb}^{(k_j)}h(\xb_0)$ are linearly independent. In fact, it turns out that only the first $n$ iterated Lie derivatives have to be computed in order to decide if such indices exist~\cite{anguelova}.

Summing up yields a simple algebraic rank criterion for local observability at a fixed point $\xb_0$, namely
\begin{equation}\label{eq:defJ}\mathrm{rank}\;\mathbf{J}:=\mathrm{rank}\lk\nabla L_{\gb}^{(0)}h(\xb_0)^T,\ldots,\nabla L_{\gb}^{(n-1)}h(\xb_0)^T\rk=n.\end{equation}
\begin{example}
Note that for a linear system
\begin{align*}
\dot{\xb}&=A\xb,\\
y&=C\xb
\end{align*}
the above condition simplifies to
\[\mathrm{rank}\lk C^T,A^T C^T,\ldots,A^{(n-1)\,T}C^T\rk=n,\]
which is the well-known criterion introduced by Kalman (see, e.g.,~\cite{sontag}). Obviously, if the linear model is locally observable at one particular initial condition, local observability holds for every other initial condition as well. 
\end{example}

To conclude this section let us assume that $\gb$ and $h$ belong to the field of rational functions in $\xb$ and $\ub=(u_1,\ldots,u_q)$ over $\R$, i.e., $\gb,h \in \R(\xb,\ub)$ with no real poles in $\mathcal{X}$. Such rational models constitute an acceptable simplification for most biological processes describing mass balance systems and including standard kinetic mechanisms (e.g., Michaelis-Menten terms). 
Consequently, the entries of the matrix $\mathbf{J}$ defined in~\eqref{eq:defJ} belong to $\R(\xb,\ub)$ and hence the rank of $\mathbf{J}(\xb,\ub)$ can be defined as the maximum size of a submatrix whose determinant is a non-zero rational function. This so-called generic rank can easily be determined by symbolic calculation software.
Hence, if the generic rank of $\mathbf{J}$ is $n$, the rank over $\R$ of $\mathbf{J}(\xb,\ub)$ 
evaluated at one specific point $(\xb,\ub) \in \R^{n+q}$ will also be $n$ except for certain singular choices that are algebraically dependent over $\R$. As these special points constitute a set of measure zero, we may thus claim that if the generic rank of $\mathbf{J}(\xb,\ub)$ equals $n$, almost every experiment corresponding to a choice of constant input variables $\ub$ will suffice to immediately distinguish a particular unknown initial condition $\xb_0$ from its neighbors by giving rise to a unique process output $y$. 

This powerful result is strengthened further by invoking the Universal Input Theorem for analytic systems due to Sussmann~\cite{sussmann1979}, ensuring that if a pair of unknown initial conditions can be distinguished by some (e.g., constant) input, it can also be distinguished by almost every input belonging to $\cinf([0,T]),\;T>0$ (see also~\cite{sontag1994} for a rigorous statement as well as a self-contained proof). Note that in the context of physiological models this is the appropriate class of functions for capturing a wide range of applicable physiological conditions.

\begin{remark}
The above line of argumentation remains valid for analytic $\gb$ and $h$. However, the generic rank in this case usually has to be estimated, e.g., by computing the rank of $\mathbf{J}$ evaluated at some randomly assigned points $(\xb,\ub) \in \R^{n+q}$.
\end{remark}

While from the above an affirmative outcome of the rank test guarantees that every initial condition within a neighborhood of the ``true'' value $\xb_0$ will yield a unique output $y$ when conducting a generic experiment with smooth inputs, the effective reconstruction of $\xb_0$ from partially observed and error-corrupted data remains a challenging problem. For instance, despite a positive result of the rank criterion, nothing can be said about the degree of linear independence between the derivatives of the output with respect to the initial conditions. In fact, the associated sensitivity matrix can still be rank-deficient in a numerical sense, indicating an \emph{ill-conditioned} problem. This is the context of practical identifiability or estimability~\cite{cobelli1980,jac1985,jac1990}, which will particularly depend on the specific experimental situation under scrutiny.

\newpage

\section{Nomenclature}
\begin{table}[H]
\centering 
\caption{Basic model parameters and nominal values during rest; LBV denotes the lean body volume in liters calculated according to $\textrm{LBV}=-16.24+0.22\,\mathrm{bh}+0.42\,\bw$, with body height ($\mathrm{bh}$) and weight ($\bw$) given in cm and kg, respectively~\cite{moerk2006}.}\label{table:param}
\begin{tabular}{|lcc|}\hline
 {\large\strut}Parameter&Symbol&Nominal value (units)\\ \hline \hline
{\large\strut}\textit{Compartment concentrations} & & \\ 
{\large\strut}	bronchioles &$\cbr$,\;$\cmeas$ & 1 ($\mu$g/l)~\cite{schwarzace}\\
{\large\strut}	alveoli &$\calv$ & \\
{\large\strut}	arterial &$\cart$ & 1 (mg/l)~\cite{wigaeus1981,kalapos2003}\\
{\large\strut}	mixed-venous &$\cven$ & \\
{\large\strut}	liver &$\cliv$ & \\
{\large\strut}	tissue &$\ctis$ & \\
{\large\strut}	inhaled (ambient) &$\cinh$ & \\
{\large\strut}\textit{Compartment volumes} & & \\
{\large\strut}	bronchioles &$V_{\mathrm{bro}}$ & 0.1 (l)~\cite{moerk2006}\\
{\large\strut}	mucosa &$V_{\mathrm{muc}}$ & 0.005 (l)~\cite{moerk2006}\\
{\large\strut}	alveoli &$V_{\mathrm{A}}$ & 4.1 (l)~\cite{moerk2006}\\
{\large\strut}	end-capillary &$V_{\mathrm{c'}}$ & 0.15 (l)~\cite{pulmcirc}\\
{\large\strut}	liver &$V_{\mathrm{liv}}$ & 0.0285\,LBV (l)~\cite{moerk2006}\\
{\large\strut}	blood liver &$V_{\mathrm{liv,b}}$ & 1.1 (l)~\cite{ottesen2004}\\
{\large\strut}	tissue &$V_{\mathrm{tis}}$ & 0.7036\,LBV (l)~\cite{moerk2006}\\
{\large\strut}\textit{Fractional blood flows at rest} & & \\
{\large\strut}	fractional flow bronchioles &$\qbr$ & 0.01~\cite{lumbbook}\\
{\large\strut}	fractional flow liver&$\qliv$ & 0.32~\cite{moerk2006}\\
{\large\strut}\textit{Partition coefficients at body temperature} & & \\
{\large\strut}	blood:air &$\hen$ & 340~\cite{anderson2006,crofford1977}\\
{\large\strut}  mucosa:air &$\lmuca$ & 392~\cite{staudinger2001,kumagai1995}\\
{\large\strut}  blood:liver &$\lliv$ & 1.73~\cite{kumagai1995}\\
{\large\strut}  blood:tissue &$\ltis$ & 1.38~\cite{anderson2006}\\
{\large\strut}\textit{Metabolic and diffusion constants} & & \\
{\large\strut}	linear metabolic rate &$\met$ & 0.0074 (l/min/kg$^{\textrm{0.75}}$) [fitted]\\
{\large\strut}  saturation metabolic rate &$\vmax$ & 0.31 (mg/min/kg$^{\textrm{0.75}}$)~\cite{kumagai1995}\\
{\large\strut}  apparent Michaelis constant &$\km$ & 84 (mg/l)~\cite{kumagai1995}\\
{\large\strut}  endogenous production &$\pr$ & 0.19 (mg/min) [fitted]\\
{\large\strut}  stratified conductance &$D$ & 0 (l/min) [fitted]\\
\hline
\end{tabular}
\end{table}

\bibliographystyle{spmpsci}
\bibliography{AcetoneBib}

\end{document}